\documentclass{amsart}
\usepackage{mathrsfs,amssymb,setspace}
\usepackage{verbatim, enumerate, color, stmaryrd, bbold, bm}

\theoremstyle{plain}
\newtheorem{theorem}{Theorem}[section]
\newtheorem{lemma}[theorem]{Lemma}
\newtheorem{proposition}[theorem]{Proposition}
\newtheorem{corollary}[theorem]{Corollary}

\theoremstyle{definition}
\newtheorem{notation}[theorem]{Notation}
\newtheorem{example}[theorem]{Example}
\newtheorem{definition}[theorem]{Definition}

\theoremstyle{remark}
\newtheorem{remark}[theorem]{Remark}

\usepackage[all]{xy}

\newcommand{\dref}[1]{Definition \ref{#1}}
\newcommand{\lref}[1]{Lemma \ref{#1}}
\newcommand{\tref}[1]{Theorem \ref{#1}}
\newcommand{\pref}[1]{Proposition \ref{#1}}
\newcommand{\cref}[1]{Corollary \ref{#1}}
\newcommand{\rref}[1]{Remark \ref{#1}}
\newcommand{\nref}[1]{Notation \ref{#1}}

\newcommand{\eref}[1]{Example \ref{#1}}
\newcommand{\sref}[1]{Section \ref{#1}}
\newcommand{\ssref}[1]{Subsection \ref{#1}}

\newcommand{\powerset}{\raisebox{.15\baselineskip}{\Large\ensuremath{\wp}}}

\begin{document}

%\doublespacing

%%-----------------------------
%%      the top matter
%%-----------------------------
\title[On the structure of $C$-algebras]{On the structure of $C$-algebras through atomicity and {\tt if-then-else}}

\author[Gayatri Panicker]{Gayatri Panicker}
\address{Department of Mathematics,  Indian Institute of Technology Guwahati, Guwahati, India}
\email{p.gayatri@iitg.ac.in}
\author[K. V. Krishna]{K. V. Krishna}
\address{Department of Mathematics, Indian Institute of Technology Guwahati, Guwahati, India}
\email{kvk@iitg.ac.in}
\author[Purandar Bhaduri]{Purandar
 Bhaduri}
\address{Department of Computer Science and Engineering, Indian Institute of Technology Guwahati, Guwahati, India}
\email{pbhaduri@iitg.ac.in}

%\date{...}

\begin{abstract}
This paper introduces the notions of atoms and atomicity in $C$-algebras and obtains a characterisation of atoms in the $C$-algebra of transformations. Further, this work presents some necessary conditions and sufficient conditions for the atomicity of $C$-algebras and shows that the class of finite atomic $C$-algebras is precisely that of finite adas. This paper also uses the {\tt if-then-else} action to study the structure of $C$-algebras and classify the elements of the $C$-algebra of transformations.
\end{abstract}



\maketitle

\section*{Introduction}

The concept of atoms in Boolean algebras is extremely useful for achieving a structural representation of Boolean algebras. When moving from two-valued Boolean logic to one that is three-valued, there are multiple such logics available depending on the interpretation of the third truth value, {\tt undefined} (e.g., see \cite{belnap70}, \cite{bergstra95}, \cite{bochvar38}, \cite{heyting34}, \cite{kleene38}, \cite{lukasiewicz20}). The three-valued logic proposed by McCarthy in \cite{mccarthy63} models the short-circuit evaluation exhibited by programming languages that evaluate expressions in sequential order, from left to right. In \cite{guzman90}, Guzm\'{a}n and Squier gave a complete axiomatization of McCarthy's three-valued logic and called the corresponding algebra a $C$-algebra, or the algebra of conditional logic. While studying {\tt if-then-else} algebras, Manes in \cite{manes93} defined an ada (algebra of disjoint alternatives) which is essentially a $C$-algebra equipped with an oracle for the halting problem. In this work, using the partial order defined by Chang in \cite{chang58} for $MV$-algebras, we adopt the notion of atoms in Boolean algebras to $C$-algebras in order to study their structure and characterise the class of finite atomic $C$-algebras.

In order to address a problem posed by Jackson and Stokes in \cite{jackson15}, present authors introduced the notion of $C$-sets and studied axiomatization of {\tt if-then-else} over $C$-algebras in \cite{panicker16,panicker17a}. Every $C$-algebra has an inbuilt {\tt if-then-else} action using which we introduce a notion of annihilators in a natural manner, which aid in studying various structural properties of $C$-algebras.

The organisation of this paper is given as follows. In \sref{SectionCAlgAdas} we recall the formal definitions of $C$-algebras and adas along with various results that will be useful to us. In \sref{SecAtomicity} we adopt the notion of atoms in Boolean algebras to $C$-algebras to study structural properties of $C$-algebras. In \ssref{SectionDef} a partial order is given on the $C$-algebra $M$, following which the notions of atoms and atomic $C$-algebras are introduced. We also state some properties related to atomicity in \ssref{SectionPropsAtoms}. On studying the $C$-algebra $\mathbb{3}^{X}$, in \ssref{SectionAtom3X} we obtain a characterisation of all atoms in $\mathbb{3}^{X}$ (cf. \tref{Thm-Atoms-3-X}), using which we establish that the $C$-algebra $\mathbb{3}^{X}$ for finite $X$ is atomic (cf. \tref{ThmFin3XAtomic}). We introduce the notion of $M$ being globally closed in $\mathbb{3}^{X}$, or g-closed in short, and observe that such finite $C$-algebras are precisely $\mathbb{3}^{X}$ (cf. \tref{ThmGClosed3X}). Subsequently, we present some necessary or sufficient conditions for the atomicity of $C$-algebras in \sref{SectionNonAtomCAlg} (cf. Theorems \ref{ThmMHashMNonAtomic}, \ref{ThmAtomless}, \ref{ThmNotAtomic}). Finally in \sref{SectionFinAtomCAlg} we obtain a characterisation of all finite atomic $C$-algebras and establish that they are precisely adas (cf. \tref{ThmAtomicAda}).

We then recall the notion of a $C$-set and of closure operators in \sref{SecCsetsandclosureop}. In \ssref{SectionAnnihilators} we introduce a notion of annihilators in $C$-algebras with $T, F, U$ through the {\tt if-then-else} action. The notion of Galois connection yields a closure operator in terms of annihilator, which in turn, yields closed sets. Further, in \sref{SectionClosedSet3X} we characterise the closed sets in the $C$-algebra of transformations $\mathbb{3}^{X}$. Additionally, we show that the collection of closed sets in $\mathbb{3}^{X}$ forms a complete Boolean algebra (cf. \tref{BooleanAlg}). Moreover, we obtain a classification of the elements of $\mathbb{3}^{X}$ where the elements of the Boolean algebra $\mathbb{2}^{X}$ form a distinct class (cf. \tref{ThmPartition}). We conclude this work with \sref{SectionChpt6Concl}.

%%%%%%%%%%%%%%%%%%%%%%%%%%%%%%%%%%%%%%%%% Chapter 1 %%%%%%%%%%%%%%%%%%%%%%%%%%%%%%%%%%%%%%%%%

\section{$C$-algebras and adas} \label{SectionCAlgAdas}

In this section we consider McCarthy's ternary logic and the algebra associated with this logic, viz., $C$-algebra, as defined by Guzm\'{a}n and Squier in \cite{guzman90}. We then present material on adas, defined by Manes in \cite{manes93}, which is a special class of $C$-algebras equipped with an oracle for the halting problem.

In \cite{kleene52}, Kleene discussed various three-valued logics that are extensions of Boolean logic. McCarthy in \cite{mccarthy63} first studied the three-valued non-commutative logic in the context of programming languages. This is the non-commutative regular extension of Boolean logic to three truth values, where the third truth value $U$ denotes the {\tt undefined} state. In this context, the evaluation of expressions is carried out sequentially from left to right, mimicking that of a majority of programming languages. A complete axiomatization for the class of algebras associated with this logic was given by Guzm\'{a}n and Squier in \cite{guzman90} and they called the algebra associated with this logic a \emph{$C$-algebra}.

\begin{definition} \label{DefCAlgebra}
 A \emph{$C$-algebra} is an algebra $\langle M, \vee, \wedge, \neg \rangle$ of
type $(2, 2, 1)$, which satisfies the following axioms for all $\alpha, \beta, \gamma \in M$:

\begin{align}
  \neg \neg \alpha & = \alpha \label{C1} \\
   \neg (\alpha \wedge \beta) & = \neg \alpha \vee \neg \beta \label{C2} \\
   (\alpha \wedge \beta) \wedge \gamma & = \alpha \wedge (\beta \wedge \gamma) \label{C3} \\
   \alpha \wedge (\beta \vee \gamma) & = (\alpha \wedge \beta) \vee (\alpha \wedge \gamma) \label{C4} \\
   (\alpha \vee \beta) \wedge \gamma & = (\alpha \wedge \gamma) \vee (\neg \alpha \wedge \beta \wedge \gamma) \label{C5} \\
   \alpha \vee (\alpha \wedge \beta) & = \alpha \label{C6} \\
   (\alpha \wedge \beta) \vee (\beta \wedge \alpha) & = (\beta \wedge \alpha) \vee (\alpha \wedge \beta) \label{C7}
\end{align}
\end{definition}

\begin{example}
Every Boolean algebra is a $C$-algebra. In particular, the two-element Boolean algebra, $\mathbb{2}$ is a $C$-algebra.
\end{example}

\begin{example}
 Let $\mathbb{3}$ denote the $C$-algebra with the universe $\{ T, F, U \}$ and the following operations. This is, in fact, McCarthy's three-valued logic. \newline
 \begin{center}
  \begin{tabular}{c|c}
  $\neg$ & \\
  \hline
  $T$ & $F$ \\
  $F$ & $T$ \\
  $U$ & $U$
 \end{tabular}
 \quad
 \begin{tabular}{c|ccc}
  $\wedge$ & $T$ & $F$ & $U$ \\
  \hline
  $T$ & $T$ & $F$ & $U$ \\
  $F$ & $F$ & $F$ & $F$ \\
  $U$ & $U$ & $U$ & $U$
 \end{tabular}
 \quad
 \begin{tabular}{c|ccc}
  $\vee$ & $T$ & $F$ & $U$ \\
  \hline
  $T$ & $T$ & $T$ & $T$ \\
  $F$ & $T$ & $F$ & $U$ \\
  $U$ & $U$ & $U$ & $U$
 \end{tabular}
 \end{center}
 \quad
\end{example}

\begin{remark} \label{RemPairsOfSets}
In view of the fact that the class of $C$-algebras is a variety, for any set $X$, $\mathbb{3}^{X}$ is a $C$-algebra with the operations defined pointwise. In fact, in \cite{guzman90} Guzm\'{a}n and Squier showed that elements of $\mathbb{3}^{X}$ along with the $C$-algebra operations may be viewed in terms of \emph{pairs of sets}. This is a pair $(A, B)$ where $A, B \subseteq X$ and $A \cap B = \emptyset$. Akin to the well-known correlation between $\mathbb{2}^{X}$ and the power set $\powerset(X)$ of $X$, for any element $\alpha \in \mathbb{3}^{X}$, associate the pair of sets $(\alpha^{-1}(T), \alpha^{-1}(F))$. Conversely, for any pair of sets $(A, B)$ where $A, B \subseteq X$ and $A \cap B = \emptyset$ associate the function $\alpha$ where $\alpha(x) = T$ if $x \in A$, $\alpha(x) = F$ if $x \in B$ and $\alpha(x) = U$ otherwise. With this correlation, the operations can be expressed as follows:
\begin{align*}
 \neg (A_{1}, A_{2}) & = (A_{2}, A_{1}) \\
 (A_{1}, A_{2}) \wedge (B_{1}, B_{2}) & = (A_{1} \cap B_{1}, A_{2} \cup (A_{1} \cap B_{2})) \\
 (A_{1}, A_{2}) \vee (B_{1}, B_{2}) & = ((A_{1} \cup (A_{2} \cap B_{1}), A_{2} \cap B_{2})
\end{align*}
\end{remark}

Further, Guzm\'{a}n and Squier showed that every $C$-algebra is a subalgebra of $\mathbb{3}^{X}$ for some $X$ as stated below.

\begin{theorem}[\cite{guzman90}] \label{SubdirIrredCAlg}
$\mathbb{3}$ and $\mathbb{2}$ are the only subdirectly irreducible $C$-algebras. Hence, every $C$-algebra is a subalgebra of a product of copies of $\mathbb{3}$.
\end{theorem}

\begin{remark}
 Considering a $C$-algebra $M$ as a subalgebra of $\mathbb{3}^{X}$, one may observe that $M_{\#} = \{ \alpha \in M : \alpha \vee \neg \alpha = T \}$ forms a Boolean algebra under the induced operations.
\end{remark}

\begin{notation} \label{NotaCAlgWithTFU}
A $C$-algebra with $T, F, U$ is a $C$-algebra with nullary operations $T, F, U$, where $T$ is the (unique) left-identity (and right-identity) for $\wedge$, $F$ is the (unique) left-identity (and right-identity) for $\vee$ and $U$ is the (unique) fixed point for $\neg$. Note that $U$ is also a left-zero for both $\wedge$ and $\vee$ while $F$ is a left-zero for $\wedge$.
\end{notation}

\begin{notation} \label{Nota3XConstants}
 The constants $T, F, U$ of the $C$-algebra $\mathbb{3}^{X}$ will be denoted by ${\bf T}, {\bf F}, {\bf U}$ respectively, and they can be identified by the pairs of sets $(X, \emptyset), (\emptyset, X), (\emptyset, \emptyset)$ respectively.

 Let $M$ be a $C$-algebra with $T, F, U$. When $M$ is considered as a subalgebra of $\mathbb{3}^{X}$, the constants $T, F, U$ of $M$ will also be denoted by ${\bf T}, {\bf F}, {\bf U}$ respectively.
\end{notation}

There is an important subclass of the variety of $C$-algebras. Manes in \cite{manes93} introduced the notion of \emph{ada} (algebra of disjoint alternatives) which is a $C$-algebra equipped with an oracle for the halting problem. He showed that the category of adas is equivalent to that of Boolean algebras. The $C$-algebra $\mathbb{3}$ is not functionally-complete. However, $\mathbb{3}$ is functionally-complete when treated as an ada. In fact, the variety of adas is generated by the ada $\mathbb{3}$.

\begin{definition}
 An \emph{ada} is a $C$-algebra $M$ with $T, F, U$ equipped with an additional unary operation $(\text{ })^{\downarrow}$ subject to the following equations for all $\alpha, \beta \in M$:
 \begin{align}
  F^{\downarrow} & = F \label{A1} \\
  U^{\downarrow} & = F \label{A2} \\
  T^{\downarrow} & = T \label{A3} \\
  \alpha \wedge \beta^{\downarrow} & = \alpha \wedge (\alpha \wedge \beta)^{\downarrow} \label{A4} \\
  \alpha^{\downarrow} \vee \neg (\alpha^{\downarrow}) & = T \label{A5} \\
  \alpha & = \alpha^{\downarrow} \vee \alpha \label{A6}
 \end{align}
\end{definition}

\begin{example} \label{Example3Ada}
 The three-element $C$-algebra $\mathbb{3}$ with the unary operation $(\text{ })^{\downarrow}$ defined as follows forms an ada.
 \begin{align*}
  T^{\downarrow} & = T \\
  U^{\downarrow} & = F = F^{\downarrow}
 \end{align*}
We also use $\mathbb{3}$ to denote this ada. One may easily resolve the notation overloading -- whether $\mathbb{3}$ is a $C$-algebra or an ada --  depending on the context.

\end{example}

In \cite{manes93}, Manes showed that the three-element ada $\mathbb{3}$ is the only subdirectly irreducible ada. For any set $X$, $\mathbb{3}^{X}$ is an ada with operations defined pointwise. Note that the three element ada $\mathbb{3}$ is also simple.

\begin{remark} \label{RemEnvAda}
 Since adas are $C$-algebras with an additional operation, every $C$-algebra $M$ freely generates an ada $\hat{M}$. That is, there exists a $C$-algebra homomorphism $\phi : M \rightarrow \hat{M}$ with the universal property that for each ada $A$ and $C$-algebra homomorphism $f: M \rightarrow A$ there exists a unique ada homomorphism $\psi: \hat{M} \rightarrow A$ with $\psi (\phi(x)) = f(x)$ for all $x \in M$. In \cite{manes93}, Manes called such an ada the \emph{enveloping ada} of $M$.
\end{remark}

Manes also showed the following result.

\begin{proposition}[\cite{manes93}]
Let $A$ be an ada. Then $A^{\downarrow} = \{ \alpha^{\downarrow} : \alpha \in A \}$ forms a Boolean algebra under the induced operations.
\end{proposition}

\begin{remark} \label{RemAdaDownarrow}
 In fact, $A^{\downarrow} = A_{\#}$. Also, $A^{\downarrow} = \{ \alpha \in A : \alpha^{\downarrow} = \alpha \}.$
\end{remark}

Further, as outlined in the following remark, Manes established that the category of adas and the category of Boolean algebras are equivalent.

\begin{remark}[\cite{manes93}] \label{RemarkStone}
Let $Q$ be a Boolean algebra. By Stone's representation of Boolean algebras, suppose $Q$ is a subalgebra of $\mathbb{2}^{X}$ for some set $X$. Consider the subalgebra $Q^{\star}$ of the ada $\mathbb{3}^{X}$ with the universe $Q^{\star} = \{(E, F) : E \cap F = \emptyset\}$ given in terms of pairs of subsets of $X$. Note that the map $Q \mapsto (Q^{\star})_{\#}$ is a Boolean isomorphism. Similarly, for an ada $A$, the map $A \mapsto (A_{\#})^{\star}$ is an ada isomorphism. Hence, the functor based on the aforesaid assignment establishes that the category of adas and the category of Boolean algebras are equivalent.
\end{remark}

\begin{remark} \label{RemFinAda3X}
 In view of the fact that the only finite Boolean algebras are $\mathbb{2}^{X}$ for finite $X$ and the equivalence of the categories of adas and Boolean algebras, we see that the only finite adas are $\mathbb{3}^{X}$ for finite $X$.
\end{remark}

\begin{notation}
Let $X$ be a set and $\bot \notin X$. The pointed set $X \cup \{ \bot \}$ with base point $\bot$ is denoted by $X_{\bot}$. The set of all functions on $X_{\bot}$ which fix $\bot$ is denoted by $\mathcal{T}_{o}(X_{\bot})$, i.e. $\mathcal{T}_{o}(X_{\bot}) = \{f \in \mathcal{T}(X_{\bot}) \; : \; f(\bot) = \bot\}$.
\end{notation}

%%%%%%%%%%%%%%%%%%%%%%%%%%%%%%%%%%%%%%%%% Chapter 7 %%%%%%%%%%%%%%%%%%%%%%%%%%%%%%%%%%%%%%%%%

\section{Atomicity} \label{SecAtomicity}

In this section we adopt the notion of atoms in Boolean algebras to $C$-algebras. First, in \ssref{SectionDef} a partial order is given on the $C$-algebra $M$, following which the notions of atoms and atomic $C$-algebras are introduced. We state various properties related to atomicity in \ssref{SectionPropsAtoms} while a characterisation of atoms in $\mathbb{3}^{X}$ is given in \ssref{SectionAtom3X} (cf. \tref{Thm-Atoms-3-X}). Subsequently, we present some necessary or sufficient conditions for the atomicity of $C$-algebras in \ssref{SectionNonAtomCAlg} (cf. Theorems \ref{ThmMHashMNonAtomic}, \ref{ThmAtomless}, \ref{ThmNotAtomic}). Finally in \ssref{SectionFinAtomCAlg} we obtain a characterisation of finite atomic $C$-algebras and establish that they are precisely adas (cf. \tref{ThmAtomicAda}).

%%%%%%%%%%%%%%%%%%%%%%%%%%%%%%%%%%%% Definitions %%%%%%%%%%%%%%%%%%%%%%%%%%%%%%%%%%%%

\subsection{Atoms and atomicity} \label{SectionDef}

We assume that $M$ is a $C$-algebra with $T, F, U$ unless mentioned otherwise. We denote elements of $M$ by $a, b, c$ and $\alpha, \beta, \gamma$. The elements of the $C$-algebra $\mathbb{3}^{X}$ will also be denoted by $\alpha, \beta, \gamma, \delta$. We continue to denote constants $T, F, U$ of $M \leq \mathbb{3}^{X}$ by ${\bf T}, {\bf F}, {\bf U}$ respectively. We begin with a partial order defined on $C$-algebras and follow the notion of the partial order given by C. C. Chang in \cite{chang58} regarding $MV$-algebras.

\begin{proposition}
 The relation $\leq$ on $M$ defined by $a \leq b$ if $a \vee b = b$ is a partial order on $M$.
\end{proposition}

\begin{proof}
 Let $a, b, c \in M$. Since $a \vee a = a$ we have $a \leq a$ from which it follows that $\leq$ is reflexive.

 Suppose that $a \leq b$ and $b \leq a$ so that $a \vee b = b$ and $b \vee a = a$. Using the fact that $M \leq \mathbb{3}^{X}$ for some set $X$ we have $a(x) \vee b(x) = b(x)$ and $b(x) \vee a(x) = a(x)$ for all $x \in X$. It suffices to consider the following three cases:
 \begin{description}
  \item [$b(x) = T$] Then $b(x) \vee a(x) = a(x)$ gives $T \vee a(x) = a(x)$ that is $a(x) = T$.
  \item [$b(x) = F$] Then $a(x) \vee b(x) = b(x)$ and so $a(x) \vee F = F$ so that $a(x) = F$.
  \item [$b(x) = U$] Then $b(x) \vee a(x) = a(x)$ that is $U \vee a(x) = a(x)$ and so $a(x) = U$.
 \end{description}
 In all three cases $a(x) = b(x)$ and so $a = b$. Hence $\leq$ is antisymmetric.

 In order to show that $\leq$ is transitive consider $a \leq b$ and $b \leq c$. Then $a \vee b = b$ and $b \vee c = c$. It is clear that $a \vee c = a \vee (b \vee c) = (a \vee b) \vee c = b \vee c = c$ and so $a \leq c$. This completes the proof.
\end{proof}

\begin{example}
 In the $C$-algebra $\mathbb{3}$ we have $F \leq T$ and $F \leq U$ while $T \nleq U$ and $U \nleq T$.
\end{example}

\begin{remark}
 In fact $F \leq a$ for all $a \in M$. This partial order does not induce a lattice structure on $M$.
\end{remark}

With this partial order we define the notion of an atom in $M$ below.

\begin{definition}
 An element $a \in M$ where $a \neq F$ is said to be an \emph{atom} if for all $b \in M$ if $F \leq b \leq a$ and $b \neq a$ then $b = F$. We denote the set of atoms of $M$ by $\mathscr{A}(M)$.

 For $A \subseteq M$ where $\{ F \} \subseteq A$ define the \emph{atoms relative to $A$} as those elements $a \in A$ such that for all $b \in A$ if $F \leq b \leq a$ and $b \neq a$ then $b = F$. We denote the set of atoms relative to $A$ as $\mathscr{A}(A)$.
\end{definition}

\begin{example}
 In $\mathbb{3}$ we have $\mathscr{A}(\mathbb{3}) = \{ T, U \}$.
\end{example}

\begin{example}
 In $\mathbb{3} \times \mathbb{3} = \mathbb{3}^{2}$ we have $\mathscr{A}(\mathbb{3}^{2}) = \{ (T, F), (F, T), (F, U), (U, F) \}$.
\end{example}

\begin{example} \label{ExampleNotUniqueRep}
 Consider $M = \mathbb{3}^{2} \setminus \{ (T, F),(F, T) \}$. Then $\mathscr{A}(M) = \{ (T, T), (F, U), (U, F) \}$.
\end{example}

\begin{remark}
 The representation of elements as join of atoms need not be unique. Consider $M = \mathbb{3}^{2} \setminus \{ (T, F),(F, T) \}$ as in \eref{ExampleNotUniqueRep}. Then $(T, T) = (T, T) \vee (F, U)$ while also $(T, T) = (T, T) \vee (U, F)$.
\end{remark}

\begin{definition}
 Let $\{ a_{i} : 1 \leq i \leq N \}$ be a finite set of atoms of $M$ such that for every rearrangement of $(a_{i})_{i = 1}^{N}$ the join of these elements remain unchanged. More precisely, if for every bijection $\sigma : \{ 1, 2, \ldots, N \} \rightarrow \{1, 2, \ldots, N \}$ we have $$a_{\sigma(1)} \vee a_{\sigma(2)} \vee \cdots \vee a_{\sigma(N)} = a_{1} \vee a_{2} \vee \cdots \vee a_{N} = a_{o} \text{ (say)}$$ then define $$\displaystyle \bigoplus_{i = 1}^{N} a_{i} = a_{o}.$$
\end{definition}

\begin{remark}
 Thus $\bigoplus_{i = 1}^{N} a_{i}$ exists when the $a_{i}$'s commute under $\vee$.
\end{remark}

\begin{example}
 Let $M = \mathbb{3}^{2}$. Then $(T, U) = (T, F) \oplus (F, U)$.
\end{example}

\begin{definition}
 Let $M$ be a $C$-algebra with $T, F, U$. We say that $M$ is \emph{atomic} if for every $(F \neq) \text{ } a \in M$ there exist a finite set of atoms $\{a_{i} : 1 \leq i \leq N \}$ such that $$a = \bigoplus_{i = 1}^{N} a_{i}$$.
\end{definition}

\begin{example}
 The $C$-algebra $M = \mathbb{3}^{2}$ is atomic.
\end{example}

\begin{example}
 Consider $M = \mathbb{3}^{2} \setminus \{ (T, F),(F, T) \}$ for which \\ $\mathscr{A}(M) = \{ (T, T), (F, U), (U, F) \}$. Then $(T, U)$ cannot be written as $\oplus$ of atoms. Thus $M$ is not atomic.
\end{example}

%%%%%%%%%%%%%%%%%%%%%%%%%%%%%%%%%% Properties of atoms %%%%%%%%%%%%%%%%%%%%%%%%%%%%%%%%%%

\subsection{Properties of atoms} \label{SectionPropsAtoms}

In this section we list some properties that are satisfied by the set of atoms of $M$.

\begin{proposition} \label{PropAtomLeq}
 Let $M$ be a finite $C$-algebra with $T, F, U$. Then for each $a \in M$ ($a \neq F$) there exists $a_{o} \in \mathscr{A}(M)$ such that $a_{o} \leq a$.
\end{proposition}

\begin{proof}
 If $a \in \mathscr{A}(M)$ then we are done since $a \leq a$. If $a \notin \mathscr{A}(M)$ then there exists $a_{1} \in M$ such that $F \lneq a_{1} \lneq a$. If $a_{1} \in \mathscr{A}(M)$ then we are done. If not, there exists $a_{2} \in M$ such that $F \lneq a_{2} \lneq a_{1} \lneq a$. Proceeding along similar lines if there is no atom in the list then there exists an infinite strictly descending chain of elements in $M$ which is a contradiction, since $M$ is finite. The result follows.
\end{proof}

This result suggests an immediate corollary. Note that $M$ is atomless if \\ $\mathscr{A}(M) = \emptyset$.

\begin{corollary} \label{CorPropAtom}
 Let $M$ be a finite $C$-algebra with $T, F, U$. Then $\mathscr{A}(M) \neq \emptyset$. Thus no finite $C$-algebra with $T, F, U$ is atomless.
\end{corollary}

The following result is concerned with the effect of the partial ordering on elements of $M_{\#}$.

\begin{proposition} \label{PropAleqB}
 If $a, b \in M$ such that $a \leq b$ and $b \in M_{\#}$ then $a \in M_{\#}$.
\end{proposition}

\begin{proof}
 Since the identity $a \vee \neg a = a \vee T$ holds in $\mathbb{3}$, it holds in all $C$-algebras. We have $a \vee b = b$ since $a \leq b$. Further, $b \in M_{\#}$ gives $b \vee \neg b = T$. Thus $a \vee \neg a = a \vee T = a \vee (b \vee \neg b) = (a \vee b) \vee \neg b = b \vee \neg b = T$ which completes the proof.
\end{proof}

We have the following corollary, which can also be proved independently.

\begin{corollary} \label{CorPropAleqB}
 If $a \in M$ such that $a \leq T$ then $a \in M_{\#}$.
\end{corollary}

We now list some properties which are useful in establishing the characterisation of atomic $C$-algebras.

\begin{proposition} \label{PropAtomRelations}
 The following hold for all $\alpha, \gamma, \delta \in M$:
 \begin{enumerate}[\rm(i)]
  \item $\alpha \wedge F \leq \alpha$.
  \item $\alpha \wedge F \leq U$.
  \item $\alpha \wedge F = U \Leftrightarrow \alpha = U$.
  \item $\alpha \wedge F = F \Leftrightarrow \alpha \in M_{\#}$.
  \item $\alpha \wedge F = \alpha \Leftrightarrow \alpha \wedge \beta = \alpha$ for all $\beta \in M$.
  \item $\alpha \leq \gamma \Rightarrow \alpha \wedge \gamma = \alpha$.
  \item $\alpha \leq \alpha \vee \beta$ for all $\beta \in M$.
  \item $\alpha \leq \delta$ and $\gamma \leq \delta \Rightarrow \alpha \vee \gamma \leq \delta$.
 \end{enumerate}
\end{proposition}

\begin{proof}$\;$
 \begin{enumerate}[(i)]
  \item In the $C$-algebra $\mathbb{3}$ consider the identity $(\alpha \wedge F) \vee \alpha = \alpha$:
  \begin{description}
   \item [$\alpha = T$] $(T \wedge F) \vee T = F \vee T = T$.
   \item [$\alpha = F$] $(F \wedge F) \vee F = F \vee F = F$.
   \item [$\alpha = U$] $(U \wedge F) \vee U = U \vee U = U$.
  \end{description}
  Thus this identity holds in $\mathbb{3}$ and so it holds in all $C$-algebras. It follows that $\alpha \wedge F \leq \alpha$.
  \item In the $C$-algebra $\mathbb{3}$ consider the identity $(\alpha \wedge F) \vee U = U$:
  \begin{description}
   \item [$\alpha = T$] $(T \wedge F) \vee U = F \vee U = U$.
   \item [$\alpha = F$] $(F \wedge F) \vee U = F \vee U = U$.
   \item [$\alpha = U$] $(U \wedge F) \vee U = U \vee U = U$.
  \end{description}
  Since this identity holds in $\mathbb{3}$ it therefore holds in all $C$-algebras. It follows that $\alpha \wedge F \leq U$.
  \item Clearly $U \wedge F = U$. Suppose that $\alpha \wedge F = U$. Since $M \leq \mathbb{3}^{X}$ for some set $X$ we have $\alpha(x) \wedge F = U$ for all $x \in X$. If $\alpha(x_{o}) \in \{ T, F \}$ for some $x_{o} \in X$ then $\alpha(x_{o}) \wedge F = F$, a contradiction. Hence $\alpha(x) = U$ for all $x \in X$ so that $\alpha = U$ in $M$.
  \item Clearly if $\alpha \in M_{\#}$ then $\alpha \wedge F = F$. Note that the identities $\alpha \wedge F = \alpha \wedge \neg \alpha$ and $\neg \alpha \vee \alpha = \alpha \vee \neg \alpha$ hold in all $C$-algebras since they hold in $\mathbb{3}$. Thus $\alpha \wedge \neg \alpha = F$. Using \eqref{C2} we have $\neg \alpha \vee \alpha = T$ so that $\alpha \vee \neg \alpha = T$. Consequently $\alpha \in M_{\#}$.
  \item It is clear that $\alpha \wedge \beta = \alpha$ for all $\beta \in M \Rightarrow \alpha \wedge F = \alpha$. Suppose that $\alpha \wedge F = \alpha$. Then for $\beta \in M$ we have $\alpha \wedge \beta = (\alpha \wedge F) \wedge \beta = \alpha \wedge (F \wedge \beta) = \alpha \wedge F = \alpha$.
  \item Since $\alpha \leq \gamma$ we have $\alpha \vee \gamma = \gamma$. Thus $\alpha \wedge \gamma = \alpha \wedge (\alpha \vee \gamma) = \alpha$ using \eqref{C1},\eqref{C2} and \eqref{C6}.
  \item Consider $\alpha \vee (\alpha \vee \beta) = (\alpha \vee \alpha) \vee \beta = \alpha \vee \beta$. Thus $\alpha \leq \alpha \vee \beta$.
  \item Consider $(\alpha \vee \gamma) \vee \delta = \alpha \vee (\gamma \vee \delta) = \alpha \vee \delta = \delta$. Thus $\alpha \vee \gamma \leq \delta$.
 \end{enumerate}
\end{proof}

\begin{remark}
 Note that the converse of \pref{PropAtomRelations}(vi) is not true in general. For instance $U \wedge F = U$, however $U \nleq F$.
\end{remark}

\begin{proposition} \label{PropAAiLeq}
 Let $a \in M$ be such that $a = \displaystyle \bigoplus_{i = 1}^{N} a_{i}$ where $a_{i} \in \mathscr{A}(M)$ for all $1\leq i \leq N$. Then $a_{i} \leq a$ for all $1\leq i \leq N$.
\end{proposition}

\begin{proof}
 Consider $a_{i} \vee a = a_{i} \vee (\bigoplus_{j = 1}^{N} a_{j}) = a_{i} \vee (a_{i} \vee \bigvee_{j \neq i} a_{j}) = (a_{i} \vee a_{i}) \vee \bigvee_{j \neq i} a_{j} = a_{i} \vee \bigvee_{j \neq i} a_{j} = \bigoplus_{j = 1}^{N} a_{j} = a$.
\end{proof}

\begin{proposition} \label{PropMMHash}
 $\mathscr{A}(M) \cap M_{\#} = \mathscr{A}(M_{\#})$. Moreover, $\mathscr{A}(M) \cap (M_{\#})^{c} \subseteq \{ a \in M : a \wedge b = a \text{ for all } b \in M \}$.
\end{proposition}

\begin{proof}
   Let $a \in \mathscr{A}(M) \cap M_{\#}$. Suppose that there exists $b \in M_{\#}$ such that $F \lneq b \lneq a$. It follows that $b \in M$ such that $F \lneq b \lneq a$ which is a contradiction to the fact that $a \in \mathscr{A}(M)$. Conversely if $a \in \mathscr{A}(M_{\#})$ then clearly $a \in M_{\#}$. If there exists $b \in M$ such that $F \lneq b \lneq a$ then using \pref{PropAleqB} we have $b \in M_{\#}$ which is a contradiction to the fact that $a \in \mathscr{A}(M_{\#})$. The result follows.

   Let $a \in \mathscr{A}(M) \cap (M_{\#})^{c}$. In order to show that $a$ is a left-zero for $\wedge$, using \pref{PropAtomRelations}(v) it suffices to show that $a \wedge F = a$. Suppose if possible that $a \wedge F \neq a$. Using \pref{PropAtomRelations}(i) we have $a \wedge F \lneq a$ and so since $a \in \mathscr{A}(M)$ it must follow that $a \wedge F = F$. Consider $M \leq \mathbb{3}^{X}$ for some set $X$. Then $a \wedge {\bf F} = {\bf F}$. If $a(x_{o}) = U$ for some $x_{o} \in X$ then $(a \wedge {\bf F})(x_{o}) = a(x_{o}) \wedge F = U \wedge F = U \neq F$, a contradiction. Thus $a(x) \in \{ T, F \}$ for all $x \in X$ and so $a \in M_{\#}$ which is a contradiction to our assumption that $a \in (M_{\#})^{c}$. Hence $a \wedge F = a$ so that $a$ is a left-zero for $\wedge$.
\end{proof}

The following result gives a necessary condition for $a$ to be an atom of $M$.

\begin{proposition} \label{PropNecAtom}
 If $a \in \mathscr{A}(M)$ then $a \wedge b \leq b \text{ or } a \wedge b = a$ for all $b \in M$.
\end{proposition}

\begin{proof}
 Let $a \in \mathscr{A}(M)$ and $b \in M$. If $a \wedge b \leq b$ then we are through. Suppose not. If $a \in \mathscr{A}(M) \cap M_{\#}^{c}$ then using \pref{PropMMHash} we have $a$ is a left-zero for $\wedge$ from which the result follows. If $a \in \mathscr{A}(M) \cap M_{\#}$ then consider $M \leq \mathbb{3}^{X}$ for some set $X$. Thus $a = a_{_{T, A}}$ for some $\emptyset \neq A \subseteq X$ so that
 \begin{equation*}
  (a \wedge b)(x) = \begin{cases}
                    b(x), & \text{ if } x \in A; \\
                    F, & \text{ otherwise.}
                    \end{cases}
 \end{equation*}
 Hence $((a \wedge b) \vee b)(x) = b(x)$ for all $x \in X$ so that $a \wedge b \leq b$.
\end{proof}

\begin{remark}
 The converse of \pref{PropNecAtom} need not be true, i.e., if $a \wedge b \leq b \text{ or } a \wedge b = a$ for all $b \in M$ then $a$ need not be in $\mathscr{A}(M)$. Consider $M = \mathbb{3}^{4}$ and $a = (U, U, F, F) \in \mathbb{3}^{4}$. This is a left-zero for $\wedge$ but is not an atom since $(F, F, F, F) \leq (U, F, F, F) \leq (U, U, F, F)$.
\end{remark}

\begin{remark}$\;$
 \begin{enumerate}[(i)]
  \item For $a \in \mathscr{A}(M)$ and $b \in M$ either $a \leq b$ or $a \wedge b \leq b$ need not, in general, hold. Consider $M = \{ (T, T, T, T), (F, F, F, F), (U, U, U, U), \\ (T, T, F, F), (F, F, T, T), (U, U, F, F), (U, U, T, T), (F, F, U, U), (T, T, U, U) \} \leq \mathbb{3}^{4}$. Take $a = (F, F, U, U) \in \mathscr{A}(M)$ and $b = (U, U, T, T) \in M$. However $a = (F, F, U, U) \nleq (U, U, T, T) = b$ and $a \wedge b = (F, F, U, U) \nleq (U, U, T, T) = b$. Note that in this case $a \wedge b = a$.
  \item For $a \in \mathscr{A}(M)$ and $b \in M$ it need not be true that $a \wedge b \in \mathscr{A}(M)$. Consider $M = \{ (T, T), (F, F), (U, U), (F, U), (T, U) \} \leq \mathbb{3}^{2}$. Take $a = (T, T) \in \mathscr{A}(M)$ and $b = (T, U) \in M$. Then $a \wedge b = (T, U) \notin \mathscr{A}(M)$.
  \item For $a, b \in M$ it need not be true that $b \leq a \vee b$. For instance in $\mathbb{3}$ we have $T \nleq U \vee T = U$.
  \item For $a, b \in M$ we need not have $a \wedge b \leq a$ nor $a \wedge b \leq b$ in general. Consider $M = \mathbb{3}^{3}$, $a = (T, U, F)$ and $b = (U, T, F)$. Then $a \wedge b = (U, U, F) \nleq (T, U, F) = a$ and $a \wedge b = (U, U, F) \nleq (U, T, F) = b$.
  \item For $a \in \mathscr{A}(M)$ it need not hold that $a \wedge U \in \mathscr{A}(M)$. Consider $M = \{ (T, T), (F, F), (U, U), (F, U), (T, U) \} \leq \mathbb{3}^{2}$ and $a = (T, T) \in \mathscr{A}(M)$ (since $\mathscr{A}(M) = \{ (T, T), (F, U) \}$). However $a \wedge {\bf U} = (T, T) \wedge (U, U) = (U, U) \notin \mathscr{A}(M)$.
 \end{enumerate}
\end{remark}

Let $\hat{M}$ be the enveloping ada of $M$ as defined in \rref{RemEnvAda}. We have the following properties in $\hat{M}$.

\begin{proposition}
 The following are equivalent for all $\beta \in M$:
 \begin{enumerate}[\rm(i)]
  \item $\beta$ is a left-zero for $\wedge$.
  \item $\beta \wedge F = \beta$.
  \item $\beta^{\downarrow} = F$ in $\hat{M}$.
 \end{enumerate}
\end{proposition}

\begin{proof}$\;$
 ((i) $\Leftrightarrow$ (ii)) This is shown in \pref{PropAtomRelations}(v). \\
 ((ii) $\Rightarrow$ (iii)) Let $\beta \wedge F = \beta$. Consider $\hat{M} \leq \mathbb{3}^{X}$ for some set $X$. Then $(\beta \wedge {\bf F})(x) = \beta(x)$ gives $\beta(x) \in \{ F, U \}$ for all $x \in X$. Thus $(\beta^{\downarrow})(x) = (\beta(x))^{\downarrow} = F$ for all $x \in X$. Hence $\beta^{\downarrow} = F$ in $\hat{M}$. \\
 ((iii)) $\Rightarrow$ (ii)) Let $\beta^{\downarrow} = F$ in $\hat{M}$. Consider $\hat{M} \leq \mathbb{3}^{X}$ for some set $X$. Then $(\beta^{\downarrow})(x) = (\beta(x))^{\downarrow} = F$ for all $x \in X$. It follows that $\beta(x) \in \{ F, U \}$ for all $x \in X$ and so $(\beta \wedge {\bf F})(x) = \beta(x)$ for all $x \in X$. Hence $\beta \wedge F = \beta$ in $M$.
\end{proof}

The left-zeros of $M$ play an important role in understanding the atomicity of $M$.

\begin{notation}
 For $\varphi \in \mathbb{3}^{X}$ denote by $\varphi_{_{T, A}}$ the element represented by the pair of sets $( A, A^{c} )$ and $\varphi_{_{U, A}}$ the element represented by the pair of sets $( \emptyset, A^{c} )$. If $A = \{ x \}$ then we simply use the notation $\varphi_{_{T, x}}$ and $\varphi_{_{U, x}}$.
\end{notation}

We now establish a relation between atoms of $M_{\#}$ and those of $M_{\#}^{c}$ for an ada $M$.

\begin{theorem}
 Let $M$ be an ada. There exists a bijection between the sets $\mathscr{A}(M) \cap M_{\#}^{c}$ and $\mathscr{A}(M) \cap M_{\#}$.
\end{theorem}

\begin{proof}
 Consider the function $G : \mathscr{A}(M) \cap M_{\#}^{c} \rightarrow \mathscr{A}(M) \cap M_{\#}$ given by the following:
 \begin{equation*}
  G(\alpha) = \neg ( (\neg \alpha)^{\downarrow}).
 \end{equation*}

 Let $\alpha \in \mathscr{A}(M) \cap M_{\#}^{c}$. It is straightforward to deduce that $G(\alpha) \in M_{\#}$. Consider $M \leq \mathbb{3}^{X}$ for some set $X$. Since $\alpha$ is a left-zero for $\wedge$ we have $\alpha = \alpha_{_{U, A}}$ for some $\emptyset \neq A \subseteq X$. It follows that $G(\alpha) = \neg ( (\neg \alpha)^{\downarrow}) = \delta_{_{T, A}}$. If $G(\alpha)$ is not an atom of $M_{\#}$ then there exists $\gamma = \gamma_{_{T, B}}$ where $\emptyset \neq B \subsetneq A$ and ${\bf F} \lneq \gamma \lneq \delta$. Thus $\beta = \gamma \wedge {\bf U} = \beta_{_{U, B}}$ and ${\bf F} \lneq \beta \lneq \alpha$ which is a contradiction to the fact that $\alpha \in \mathscr{A}(M) \cap M_{\#}^{c}$. It follows that $G$ is well-defined.

 Suppose that $\neg ( (\neg \alpha)^{\downarrow}) = \neg ( (\neg \beta)^{\downarrow})$ for some $\alpha, \beta \in \mathscr{A}(M) \cap M_{\#}^{c}$. Then $(\neg \alpha)^{\downarrow} = (\neg \beta)^{\downarrow} \in M_{\#}$. Consider $M \leq \mathbb{3}^{X}$ for some set $X$. Then $(\neg \alpha)^{\downarrow} = (\neg \beta)^{\downarrow} = \gamma_{_{T, A}}$ for some $A \subseteq X$. It follows that $\neg \alpha$ and $\neg \beta$ can be represented by the pairs of sets $(A, B_{\alpha})$ and $(A, B_{\beta})$ where $B_{\alpha}, B_{\beta} \subseteq A^{c}$. Thus $\alpha$ and $\beta$ can be represented by the pairs of sets $(B_{\alpha}, A)$ and $(B_{\beta}, A)$ where $B_{\alpha}, B_{\beta} \subseteq A^{c}$. Since $\alpha, \beta \in \mathscr{A}(M) \cap M_{\#}^{c}$ we have $\alpha = \alpha_{_{U, C}}$ and $\beta = \beta_{_{U, D}}$ for some $C, D \subseteq X$. Hence in the representation for $\alpha$ and $\beta$ that is $(B_{\alpha}, A)$ and $(B_{\beta}, A)$ respectively we must have $B_{\alpha} = \emptyset = B_{\beta}$. It follows that $\alpha = \beta$ and so $G$ is injective.

 Let $\beta \in \mathscr{A}(M) \cap M_{\#}$. Consider $M \leq \mathbb{3}^{X}$ for some set $X$. It follows that $\beta = \beta_{_{T, A}}$ for some $\emptyset \neq A \subseteq X$. Consider $\alpha = \beta \wedge {\bf U} = \alpha_{_{U, A}} \in M_{\#}^{c}$. Along similar lines as in the proof for the well-definedness of $G$, we show that $\alpha \in \mathscr{A}(M) \cap M_{\#}^{c}$. Further, $G(\alpha) = \beta$ so that $G$ is surjective.

\end{proof}

\begin{corollary} \label{CorAdaAtomEven}
 Let $M$ be a finite ada. Then $| \mathscr{A}(M) |$ is even.
\end{corollary}

%%%%%%%%%%%%%%%%%%%%%%%%%%%%%%%%%% Atomicity of 3^X %%%%%%%%%%%%%%%%%%%%%%%%%%%%%%%%%%

\subsection{Atomicity of $\mathbb{3}^{X}$} \label{SectionAtom3X}

We consider the $C$-algebra $\mathbb{3}^{X}$ and first establish a characterisation for its atoms.

\begin{theorem} \label{Thm-Atoms-3-X}
 Let $X$ be any set. Then $\mathscr{A}(\mathbb{3}^{X}) = \{ \alpha \in \mathbb{3}^{X} : \text{ there exists a unique } \\ x_{o} \in X \text{ such that } \alpha(x_{o}) \in \{ T, U \} \}$.
\end{theorem}

\begin{proof}
 Let $M = \mathbb{3}^{X}$ and $A = \{ \alpha \in M : \text{ there exists a unique } x_{o} \in X \text{ such that } \\ \alpha(x_{o}) \in \{ T, U \} \}$. Let $\alpha \in A$. If $\alpha$ is not an atom then there exists $\beta \in M$ such that ${\bf F} \lneq \beta \lneq \alpha$. Since $\alpha \in A$ we have $\alpha(x) = F$ for all $x \neq x_{o}$. Thus since $F \leq \beta(x) \leq \alpha(x)$ we must have $\beta(x) = F$ for all $x \neq x_{o}$. It is clear that since $\beta \lneq \alpha$ we must have $\beta(x_{o}) \lneq \alpha(x_{o})$ and so $\beta(x_{o}) = F$. This holds as $F \lneq T$ and $F \lneq U$ but $T \nleq U$ and $U \nleq T$. However this gives $\beta = {\bf F}$ which is a contradiction. Thus $\alpha \in \mathscr{A}(M)$.

 Conversely suppose that $\alpha \in \mathscr{A}(M)$ but $\alpha \notin A$. Then either there exist $x_{o}, y_{o} \in X$ where $x_{o} \neq y_{o}$ and $\alpha(x_{o}), \alpha(y_{o}) \in \{ T, U \}$ or we have $\alpha(x) = F$ for all $x \in X$. If $\alpha(x) = F$ for all $x \in X$ then clearly $\alpha = {\bf F}$ and so $\alpha \notin \mathscr{A}(M)$ which is a contradiction. If there exist $x_{o}, y_{o} \in X$ where $x_{o} \neq y_{o}$ and $\alpha(x_{o}), \alpha(y_{o}) \in \{ T, U \}$ then consider $\beta \in M$ given by the following:
 \begin{equation*}
  \beta(x) = \begin{cases}
              \alpha(x), & \text{ if } x \neq x_{o}; \\
              F, & \text{ if } x = x_{o}.
             \end{cases}
 \end{equation*}
 It is easy to see that $F \leq \beta(x) \leq \alpha(x)$ for all $x \in X$ and so ${\bf F} \leq \beta \leq \alpha$. Since $\beta(x_{o}) = F \lneq \alpha(x_{o})$ and $\beta(y_{o}) = \alpha(y_{o}) \in \{ T, U \}$ we have ${\bf F} \lneq \beta \lneq \alpha$ which is a contradiction to the assumption that $\alpha \in \mathscr{A}(M)$. The result follows.
\end{proof}

This gives us the following result on the number of atoms in $\mathbb{3}^{X}$.

\begin{corollary} \label{Cor-Atoms-3-X}
 For $X \neq \emptyset$ we have $| \mathscr{A}(\mathbb{3}^{X}) | = 2 \times | X |$.
\end{corollary}

In view of the fact that all finite adas are isomorphic to $\mathbb{3}^{X}$ (cf. \rref{RemFinAda3X}) we note that \cref{Cor-Atoms-3-X} is in fact a stronger version of \cref{CorAdaAtomEven}.

We now study the set of atoms in $\mathbb{3}^{X}$ that have existence of $\oplus$.

\begin{notation}
 Let $\alpha \in \mathscr{A}(\mathbb{3}^{X})$. Using \tref{Thm-Atoms-3-X}, denote by $x_{\alpha}$ the unique co-ordinate satisfying $\alpha(x_{\alpha}) \in \{ T, U \}$.
\end{notation}

\begin{theorem} \label{ThmExisVeeBar3X}
 Let $\{ \alpha_{i}: 1 \leq i \leq N \}$ be a finite set of atoms in $\mathbb{3}^{X}$. Then $\displaystyle \bigoplus_{i = 1}^{N} \alpha_{i}$ exists if and only if $x_{\alpha_{i}} \neq x_{\alpha_{j}}$ for all $i \neq j$. Further
 \begin{equation*}
  \bigoplus_{i = 1}^{N} \alpha_{i}(x) = \begin{cases}
                                        \alpha_{i}(x_{\alpha_{i}}), & \text{ if } x = x_{\alpha_{i}}; \\
                                        F, & \text{ otherwise.}
                                        \end{cases}
 \end{equation*}
\end{theorem}

\begin{proof}
 Let $x_{\alpha_{i}} \neq x_{\alpha_{j}}$ for all $i \neq j$. Thus using \tref{Thm-Atoms-3-X} for any $x \in X$ there exist at most one $\alpha_{x}$ in this collection such that $\alpha_{x}(x) \in \{ T, U \}$. In view of the fact that $F$ is a left and right-identity for $\vee$ we have
 \begin{equation*}
  \alpha_{1}(x) \vee \alpha_{2}(x) \vee \cdots \vee \alpha_{N}(x) = \begin{cases}
                                                                \alpha_{i}(x_{\alpha_{i}}), & \text{ if } x = x_{\alpha_{i}} \text{ for some } 1  \leq i \leq N; \\
                                                                F, & \text{ otherwise.}
                                                               \end{cases}
 \end{equation*}
 Hence for any bijection $\sigma : \{ 1, 2, \ldots, N \} \rightarrow \{ 1, 2, \ldots, N \}$ we have:
 \begin{equation*}
  \alpha_{\sigma(1)}(x) \vee \alpha_{\sigma(2)}(x) \vee \cdots \vee \alpha_{\sigma(N)}(x) = \begin{cases}
                                                                                             \alpha_{i}(x_{\alpha_{i}}), & \text{ if } x = x_{\alpha_{i}} \text{ for some } 1  \leq i \leq N; \\
                                                                                             F, & \text{ otherwise.}
                                                                                             \end{cases}
 \end{equation*}
  It follows that $\alpha_{\sigma(1)} \vee \alpha_{\sigma(2)} \vee \cdots \vee \alpha_{\sigma(N)} = \alpha_{1} \vee \alpha_{2} \vee \cdots \vee \alpha_{N}$ and so $\displaystyle \bigoplus_{i = 1}^{N} \alpha_{i}$ exists.

 Conversely suppose that $\displaystyle \bigoplus_{i = 1}^{N} \alpha_{i}$ exists and $x_{\alpha_{i}} = x_{\alpha_{j}}$ for some $i \neq j$. Without loss of generality assume that $\alpha_{i}(x_{\alpha_{i}}) = T$ while $\alpha_{j}(x_{\alpha_{i}}) = U$. It follows that for the bijection $\sigma : \{ 1, 2, \ldots, N \} \rightarrow \{ 1, 2, \ldots, N \}$ given by
 \begin{equation*}
  \sigma(n) = \begin{cases}
              i, & \text{ if } n = 1; \\
              j, & \text{ if } n = 2; \\
              n, & \text{ otherwise}
              \end{cases}
 \end{equation*}
 we have $\alpha_{\sigma(1)}(x_{\alpha_{i}}) \vee \alpha_{\sigma(2)}(x_{\alpha_{i}}) \vee \cdots \vee \alpha_{\sigma(N)}(x_{\alpha_{i}}) = \alpha_{i}(x_{\alpha_{i}}) \vee \alpha_{j}(x_{{o}_{\alpha_{i}}}) \vee \cdots \vee \alpha_{\sigma(N)}(x_{{o}_{\alpha_{i}}}) = T \vee U \vee \cdots \vee \alpha_{\sigma_{1}(N)}(x_{\alpha_{i}}) = T$, since $T$ is a left-zero for $\vee$.

 On the other hand for the bijection $\tau : \{ 1, 2, \ldots, N \} \rightarrow \{ 1, 2, \ldots, N \}$ where
 \begin{equation*}
  \tau(n) = \begin{cases}
              j, & \text{ if } n = 1; \\
              i, & \text{ if } n = 2; \\
              n, & \text{ otherwise}
              \end{cases}
 \end{equation*}
 we have $\alpha_{\tau(1)}(x_{\alpha_{i}}) \vee \alpha_{\tau(2)}(x_{\alpha_{i}}) \vee \cdots \vee \alpha_{\tau(N)}(x_{\alpha_{i}}) = \alpha_{j}(x_{\alpha_{i}}) \vee \alpha_{i}(x_{\alpha_{i}}) \vee \cdots \vee \alpha_{\tau(N)}(x_{\alpha_{i}}) = U \vee T \vee \cdots \vee \alpha_{\tau(N)}(x_{\alpha_{i}}) = U$, since $U$ is a left-zero for $\vee$. This is a contradiction to the assumption that $\bigoplus_{i = 1}^{N} \alpha_{i}$ exists and the result follows. The expression for $\bigoplus \alpha_{i}$ is also clear from the above proof.
\end{proof}

\begin{theorem} \label{ThmFin3XAtomic}
 If $X$ is finite then $\mathbb{3}^{X}$ is atomic.
\end{theorem}

\begin{proof}
 Let $\beta \in \mathbb{3}^{X}$ such that $\beta \neq {\bf F}$. Using the pairs of sets representation of $\mathbb{3}^{X}$ identify $\beta$ with the pair of sets $(A, B)$. Since $\beta \neq {\bf F}$ it follows that $B^{c} \neq \emptyset$. Consider the family of elements defined by the following for $y \in B^{c}$:
 \begin{equation*}
  \alpha_{y}(x) = \begin{cases}
                   \beta(y), & \text{ if } x = y; \\
                   F, & \text{ otherwise.}
                  \end{cases}
 \end{equation*}
 Using \tref{Thm-Atoms-3-X} since $\alpha_{y}(y) = \beta(y) \in \{ T, U \}$ we have $\alpha_{y} \in \mathscr{A}(\mathbb{3}^{X})$ for each $y \in B^{c}$. Further there are finitely many $\alpha_{y}$ since $X$ and therefore $B^{c}$ is finite. Note that $x_{\alpha_{y}} = y$ and so $x_{\alpha_{y}} \neq x_{\alpha_{z}}$ for $y \neq z$. Consequently, using \tref{ThmExisVeeBar3X} $\displaystyle \bigoplus_{y \in B^{c}} \alpha_{y}$ exists.

 For $x \in B$ we have $\beta(x) = F$. Also $\alpha_{y}(x) = F$ for all $y \in B^{c}$ and so $\bigoplus \alpha_{y}(x) = F = \beta(x)$. For $x \in B^{c}$ using \tref{ThmExisVeeBar3X} we have $\bigoplus \alpha_{y}(x) = \alpha_{x}(x) = \beta(x)$. Thus we have a finite set $\{ \alpha_{y} : y \in B^{c} \} \subseteq \mathscr{A}(\mathbb{3}^{X})$ such that $\bigoplus \alpha_{y} = \beta$. Hence $\mathbb{3}^{X}$ is atomic.
\end{proof}

\begin{remark}
 Note that $\mathbb{3}^{X}$ where $X$ is infinite will be non-atomic since the element ${\bf T}$ can never be expressed in terms of finitely many atoms.
\end{remark}

%%%%%%%%%%%%%%%%%%%%%%%%%%%%%%%%%% g-closed $C$-algebras %%%%%%%%%%%%%%%%%%%%%%%%%%%%%%%%%%

\subsection{g-closed $C$-algebras} \label{SectionGClosed}

We consider $M \leq \mathbb{3}^{X}$ and try to understand the atomicity of $M$ from information about the atoms of $\mathbb{3}^{X}$. First we justify the feasibility of this approach.

\begin{remark}
 Let $\phi : M \rightarrow \mathbb{3}^{X}$ be a $C$-algebra embedding. Then $\phi$ is also order-preserving. Let $x \leq y \in M$. Then $\phi(x) \vee \phi(y) = \phi(x \vee y) = \phi(y)$ and so $\phi(x) \leq \phi(y)$.
\end{remark}

Thus we make use of the notion of atoms in $\mathbb{3}^{X}$ to gain an understanding of the same in $M$ where $M \leq \mathbb{3}^{X}$. In this section we assume that $M \leq \mathbb{3}^{X}$.

\begin{remark}
 It is straightforward to verify that $M \cap \mathscr{A}(\mathbb{3}^{X}) \subseteq \mathscr{A}(M)$. In general the inclusion could be proper.

 To illustrate this consider $$M = \{ (T, T), (F, F), (U, U), (F, U), (U, F), (T, U), (U, T) \}$$ where $M \leq \mathbb{3}^{2}$. Then $\mathscr{A}(M) = \{ (F, U), (U, F), (T, T) \} \supsetneq M \cap \mathscr{A}(\mathbb{3}^{2})$ since $(T, T) \notin \mathscr{A}(\mathbb{3}^{2})$.
\end{remark}

Thus not all atoms of $M$ are atoms of $\mathbb{3}^{X}$. The atoms of $M$ that remain atoms in $\mathbb{3}^{X}$ are in some sense \emph{global} atoms. If every atom of $M$ is an atom of $\mathbb{3}^{X}$, and therefore, a global atom, then $M$ is closed with respect to global atoms. Thus we define the following notion.

\begin{definition}
  $M$ is said to be \emph{closed with respect to global atoms in} $\mathbb{3}^{X}$ if $\mathscr{A}(M) \subseteq \mathscr{A}(\mathbb{3}^{X})$. In short we say that $M$ is \emph{g-closed in} $\mathbb{3}^{X}$.
\end{definition}

\begin{remark}
 If $M$ is g-closed in $\mathbb{3}^{X}$ then we have $\mathscr{A}(M) = M \cap \mathscr{A}(\mathbb{3}^{X})$.
\end{remark}

\begin{remark} \label{RemGClosed32}
 Consider $M \leq \mathbb{3}^{2}$. The subalgebras of $\mathbb{3}^{2}$ are as follows:
 \begin{align*}
  M_{0} & = \{ (T, T), (F, F), (U, U) \}, \\
  M_{1} & = \{ (T, T), (F, F), (U, U), (F, U), (T, U) \}, \\
  M_{2} & = \{ (T, T), (F, F), (U, U), (U, F), (U, T) \}, \\
  M_{3} & = \{ (T, T), (F, F), (U, U), (F, U), (T, U), (U, F), (U, T) \}, \\
  M_{4} & = \mathbb{3}^{2}.
 \end{align*}
 The set of atoms of each subalgebra are as follows:
 \begin{align*}
  \mathscr{A}(M_{0}) & = \{ (T, T), (U, U) \}, \\
  \mathscr{A}(M_{1}) & = \{ (T, T), (F, U) \}, \\
  \mathscr{A}(M_{2}) & = \{ (T, T), (U, F) \}, \\
  \mathscr{A}(M_{3}) & = \{ (T, T), (F, U), (U, F) \}, \\
  \mathscr{A}(M_{4}) & = \mathbb{3}^{2}.
 \end{align*}
 Since $(T, T) \in \mathscr{A}(M_{i})$ for $0 \leq i \leq 3$ and $(T, T) \notin \mathscr{A}(\mathbb{3}^{2})$, no proper subalgebra is g-closed in $\mathbb{3}^{2}$.
\end{remark}

We ascertain all the globally closed subalgebras of $\mathbb{3}^{X}$. To that aim we first have the following result.

\begin{lemma} \label{LemmaAtomsEqual3X}
 Let $M \leq \mathbb{3}^{X}$ where $X$ is finite. If $\mathscr{A}(M) = \mathscr{A}(\mathbb{3}^{X})$ then $M = \mathbb{3}^{X}$.
\end{lemma}

\begin{proof}
 Let $\alpha \in \mathbb{3}^{X}$. If $\alpha = {\bf F}$ then we are done since ${\bf F} \in M$. Suppose that $\alpha \neq {\bf F}$. Then $\alpha$ can be represented by the pair of sets $(A, B)$ where $B^{c} \neq \emptyset$. Consider as earlier for each $y \in B^{c}$ the family of elements given below:
 \begin{equation*}
  \alpha_{y}(x) = \begin{cases}
                   \alpha(y), & \text{ if } x = y; \\
                   F, & \text{ otherwise.}
                  \end{cases}
 \end{equation*}
 Using \tref{Thm-Atoms-3-X} we have $\alpha_{y} \in \mathscr{A}(\mathbb{3}^{X})$. Further $\mathscr{A}(M) = \mathscr{A}(\mathbb{3}^{X})$ gives $\alpha_{y} \in \mathscr{A}(M) \subseteq M$. Note that since $X$ is finite, so is $B^{c}$. Consequently there are only finitely many such $\alpha_{y}$. Moreover using \tref{ThmExisVeeBar3X} $\bigoplus \alpha_{y}$ exists and so $\bigoplus \alpha_{y} \in \mathbb{3}^{X}$ so that $\bigoplus \alpha_{y} \in M$. It is straightforward to verify that $\displaystyle \bigoplus_{y \in B^{c}} \alpha_{y} = \alpha$ so that $\alpha \in M$. Thus $M = \mathbb{3}^{X}$.
\end{proof}

\begin{theorem} \label{ThmGClosed3X}
 Let $M$ be g-closed in $\mathbb{3}^{X}$ where $X$ is finite. Then $M = \mathbb{3}^{X}$.
\end{theorem}

\begin{proof}
 We describe an algorithmic mechanism to generate all atoms from one. In view of \lref{LemmaAtomsEqual3X}, on obtaining $\mathscr{A}(M) = \mathscr{A}(\mathbb{3}^{X})$ we then have $M = \mathbb{3}^{X}$. It suffices to show that $\alpha_{_{T, x}} \in M$ for each $x \in X$, because if $\alpha_{_{T, x}} \in M$ then $\alpha_{_{T, x}} \wedge {\bf U } = \alpha_{_{U, x}} \in M$.

 Since $X$ is finite we have $M$ is finite. Using \pref{PropAtomLeq} and \pref{PropAleqB} for ${\bf T} \in M$ there exists $\alpha \in \mathscr{A}(M)$ such that $\alpha \leq {\bf T}$ so that $\alpha \in M_{\#} \leq \mathbb{2}^{X}$. Since $M$ is g-closed in $\mathbb{3}^{X}$ we have $\alpha \in \mathscr{A}(\mathbb{3}^{X})$ so that $\alpha = \alpha_{_{T, x_{1}}}$ for some $x_{1} \in X$.

 Define $\beta_{1} = \alpha_{_{T, x_{1}}}$ and so $\neg \beta_{1} = \neg \alpha_{_{T, x_{1}}} = \alpha_{_{T, X \setminus \{ x_{1} \} }}$. If $\neg \beta_{1}$ is an atom then $ X \setminus \{ x_{1} \}$ is a singleton and so $X = \{ x_{1}, x_{2} \}$ and so the algebra is $\mathbb{3}^{2}$. The only subalgebra g-closed in $\mathbb{3}^{2}$ is itself and we are done.

 If $\neg \beta_{1}$ is not an atom then there exists $\alpha_{_{T, x_{2}}} \in \mathscr{A}(M)$ such that $\alpha_{_{T, x_{2}}} \leq \neg \beta_{1} \leq {\bf T}$. Define $\beta_{2} = \alpha_{_{T, x_{2}}}$ and so $\neg \beta_{2} = \neg \alpha_{_{T, x_{2}}} = \alpha_{_{T, X \setminus \{ x_{2} \} }}$. If $\neg \beta_{2}$ is an atom then we are through. Else there exists $\alpha_{_{T, x_{3}}} \in \mathscr{A}(M)$ such that $\alpha_{_{T, x_{3}}} \leq \neg \beta_{2} \leq {\bf T}$. Define $\beta_{3} = \alpha_{_{T, x_{3}}}$ and so $\neg \beta_{3} = \neg \alpha_{_{T, x_{3}}} = \alpha_{_{T, X \setminus \{ x_{3} \} }}$ and so on.

 This process can take at most $| X |$ steps. Further, as mentioned above if $\alpha_{_{T, x}} \in M$ then $\alpha_{_{T, x}} \wedge {\bf U } = \alpha_{_{U, x}} \in M$ so that $\mathscr{A}(M) = \mathscr{A}(\mathbb{3}^{X})$. Hence $M = \mathbb{3}^{X}$.
\end{proof}

\begin{corollary}
 The collection of g-closed subalgebras in $\mathbb{3}^{X}$ where $X$ is finite comprises atomic algebras.
\end{corollary}

\begin{remark}$\;$
 \begin{enumerate}[(i)]
  \item Let $M \leq \mathbb{3}^{X}$ where $X$ is finite and $M$ is not g-closed in $\mathbb{3}^{X}$. Then $M$ may be atomic. For instance, consider $M = M_{0} \leq \mathbb{3}^{2}$ as given in \rref{RemGClosed32}. We know that $M_{0}$ is not g-closed in $\mathbb{3}^{2}$. However $M_{0}$ is clearly atomic.
  \item Let $M \leq \mathbb{3}^{X}$ where $X$ is finite, $M$ is non-trivial and $M$ is not g-closed in $\mathbb{3}^{X}$. Then $M$ may still be atomic. Consider $M = \{ (T, T, T, T), (F, F, F, F), (U, U, U, U), \\ (T, T, F, F), (F, F, T, T), (U, U, F, F), (U, U, T, T), (F, F, U, U), (T, T, U, U) \} \leq \mathbb{3}^{4}$. In this case $\mathscr{A}(M) = \{ (T, T, F, F), (F, F, T, T), (U, U, F, F), (F, F, U, U) \}$ and so it is not g-closed in $\mathbb{3}^{4}$. However $M$ is atomic.
  \end{enumerate}
\end{remark}

%%%%%%%%%%%%%%%%%%%%%%%%%% Non-atomic $C$-algebras %%%%%%%%%%%%%%%%%%%%%%%%%%

\subsection{Non-atomic $C$-algebras} \label{SectionNonAtomCAlg}

We now investigate the relation between the atomicity of $M_{\#}$ and that of $M$. It is a straightforward assertion that if $M$ is a finite $C$-algebra with $T, F, U$ then no such relation holds since $M_{\#}$ is always atomic but $M$ need not be so. However the question stands in the case where $M$ is infinite. In this section we consider $M$ to be a $C$-algebra with $T, F, U$ unless otherwise mentioned.

\begin{theorem} \label{ThmMHashMNonAtomic}
 If $M_{\#}$ is non-atomic then $M$ is non-atomic.
\end{theorem}

\begin{proof}
 If possible let $M_{\#}$ be non-atomic and $M$ be atomic. Let $a \in M_{\#} \subseteq M$ then there exist finitely many $a_{i} \in \mathscr{A}(M)$ such that $a = \bigoplus a_{i}$. Using \pref{PropAAiLeq} we have $a_{i} \leq a$. Moreover, using \pref{PropAleqB} we have $a_{i} \in M_{\#}$. Further using \pref{PropMMHash} we have $\mathscr{A}(M) \cap M_{\#} = \mathscr{A}(M_{\#})$ so that $a_{i} \in \mathscr{A}(M_{\#})$ and $a = \bigoplus a_{i}$. Thus $M_{\#}$ is atomic, a contradiction.
\end{proof}

The following result relates to atomless adas.

\begin{theorem} \label{ThmAtomless}
 Let $M$ be an ada. If $M_{\#}$ is atomless then $M$ is atomless.
\end{theorem}

\begin{proof}
 If possible let $M_{\#}$ be atomless but $M$ not be atomless. Therefore let $\alpha \in \mathscr{A}(M)$. It is clear that $\alpha \notin M_{\#}$ since otherwise using \pref{PropMMHash} we have $\alpha \in \mathscr{A}(M) \cap M_{\#} = \mathscr{A}(M_{\#})$ which is a contradiction since $M_{\#}$ is atomless. Thus $\alpha \in M_{\#}^{c}$ and so $\alpha^{\downarrow} \neq \alpha$ (cf. \rref{RemAdaDownarrow}). We have the following cases.

 \emph{Case I}: $\alpha^{\downarrow} \neq F$: The ada identity $\alpha^{\downarrow} \vee \alpha = \alpha$ holds in $\mathbb{3}$ and therefore in all adas. Thus we have $F \lneq \alpha^{\downarrow} \lneq \alpha$ which is a contradiction since $\alpha \in \mathscr{A}(M)$.

 \emph{Case II}: $\alpha^{\downarrow} = F$: Using \pref{PropMMHash} we have $\alpha \in \{ a \in M : a \wedge b = a \text{ for all } b \in M \}$. Consider $M \leq \mathbb{3}^{X}$ for some set $X$. It follows that $\alpha = \alpha_{_{U, A}}$ for some $A \subseteq X$. This is true since if $\alpha(x) = T$ for some $x \in X$ then $\alpha^{\downarrow}(x) = T$ and so $\alpha^{\downarrow} \neq {\bf F}$. Also $A \neq \emptyset$ since $\alpha \neq {\bf F}$. Then
 \begin{equation*}
  \neg \alpha(x) = \neg \alpha_{_{U, A}}(x) = \begin{cases}
                                               U, & \text{ if } x \in A; \\
                                               T, & \text{ otherwise.}
                                              \end{cases}
 \end{equation*}
 Then $(\neg \alpha)^{\downarrow} \in M$ since $M$ is an ada so that
 \begin{equation*}
  (\neg \alpha)^{\downarrow}(x) = \begin{cases}
                                   F, & \text{ if } x \in A; \\
                                   T, & \text{ otherwise.}
                                   \end{cases}
 \end{equation*}
 In fact $(\neg \alpha)^{\downarrow} \in M_{\#}$. Consider $\neg ((\neg \alpha)^{\downarrow}) \in M_{\#}$ where in fact $\neg ((\neg \alpha)^{\downarrow}) = \beta_{_{T, A}} \in M_{\#}$. Since $M_{\#}$ is atomless it follows that there exists $\beta_{_{T, B}} \in M_{\#}$ where $\emptyset \neq B \subsetneq A$ and ${\bf F} \lneq \beta_{_{T, B}} \lneq \beta_{_{T, A}}$. Consider $\beta_{_{U, B}} = \beta_{_{T, B}} \wedge {\bf U} \in M$. Since $\emptyset \neq B \subsetneq A$ we have ${\bf F} \lneq \beta_{_{U, B}} \lneq \alpha_{_{U, A}} = \alpha$ which is a contradiction to the fact that $\alpha \in \mathscr{A}(M)$.
\end{proof}

\begin{remark}
 \tref{ThmAtomless} allows us to  construct an atomless ada from an atomless Boolean algebra. For an atomless Boolean algebra $B$, the ada $B^{\star}$ will also be atomless. For further reading on atomless Boolean algebras refer to \cite{halmos09}.
\end{remark}

\begin{theorem} \label{ThmNotAtomic}
 Let $M$ be a finite $C$-algebra with $T, F, U$ such that $|M| > 3$ and $T \in \mathscr{A}(M)$. Then $M$ is not atomic.
\end{theorem}

\begin{proof}
 Since $T \in \mathscr{A}(M)$ it is clear that $M_{\#} = \{ T, F \}$. Since $|M| > 3$ there exists $\gamma \in M \setminus \{ T, F, U \}$ and since $M$ is finite, using \pref{PropAtomLeq} there exists $\alpha \in \mathscr{A}(M)$ such that $\alpha \leq \gamma$. Clearly $\alpha \in \mathscr{A}(M) \cap M_{\#}^{c}$.

 Consider $M \leq \mathbb{3}^{X}$ for some set $X$. Then $\alpha = \alpha_{_{U, A}}$ for some $\emptyset \neq A \subseteq X$. Suppose that $A = X$. Then $\alpha = {\bf U} \in \mathscr{A}(M)$. Hence $M = \{ {\bf T}, {\bf F}, {\bf U} \}$ else if there was some $\beta \in M_{\#}^{c} \setminus \{ {\bf U} \}$ then using \pref{PropAtomRelations}(ii) we have ${\bf F} \lneq \beta \wedge F \lneq {\bf U}$ which is a contradiction to the fact that ${\bf U} \in \mathscr{A}(M)$. Thus $M = \{ {\bf T}, {\bf F}, {\bf U} \}$, a contradiction to our assumption that $M$ is non-trivial. Thus $\alpha = \alpha_{_{U, A}}$ where $\emptyset \neq A \subsetneq X$.

 Suppose that $M$ is atomic. Consider $\neg \alpha \in M_{\#}^{c}$. There exist finitely many $a_{i} \in \mathscr{A}(M)$ such that $\neg \alpha = \bigoplus a_{i}$. Clearly ${\bf T} \notin \{ a_{i} \}$ since ${\bf T} \vee a = {\bf T} \neq \neg \alpha$. Since $M_{\#} = \{ {\bf T}, {\bf F} \}$ we have $\mathscr{A}(M) \setminus \{ {\bf T} \} \subseteq M_{\#}^{c}$. Thus $a_{i} = a_{_{U, A_{i}}}$ for $\emptyset \neq A_{i} \subseteq X$. However $\neg \alpha = \neg \alpha_{_{U, A}}$ where $\emptyset \neq A \subsetneq X$ and so we have
 \begin{equation*}
  \neg \alpha(x) = \begin{cases}
               U, & \text{ if } x \in A; \\
               T, & \text{ otherwise.}
              \end{cases}
 \end{equation*}
 Moreover $\neg \alpha = \bigoplus a_{_{U, A_{i}}}$ gives
 \begin{equation*}
  \neg \alpha(x) = \begin{cases}
               U, & \text{ if } x \in A_{i} \text{ for some i}; \\
               F, & \text{ otherwise.}
              \end{cases}
 \end{equation*}
 This is a contradiction since $A \subsetneq X$ which implies that $\neg \alpha(x_{o}) = T$ for some $x_{o} \in X$. Hence $M$ is not atomic.
\end{proof}

\begin{corollary}
 Let $M$ be a finite $C$-algebra with $T, F, U$ such that $|M| > 3$. Then $\overline{M_{\#}^{c}} = M_{\#}^{c} \cup \{ T, F \}$ is never atomic.
\end{corollary}

\begin{proof}
 Since $(\overline{M_{\#}^{c}})_{\#} = \{ T, F \}$ we have $T \in \mathscr{A}(\overline{M_{\#}^{c}})$. The result follows from \tref{ThmNotAtomic}.
\end{proof}

\begin{remark}
 The converse of \tref{ThmNotAtomic} need not be true. That is, if $M$ be a $C$-algebra with $T, F, U$ such that $M$ is not atomic then $T$ need not be in $\mathscr{A}(M)$. Consider $M = \{ (T, T, T), (F, F, F), (U, U, U), (U, F, F), (U, T, T), (F, F, T), (T, T, F), \\ (F, F, U), (T, T, U), (U, U, F), (U, T, F), (U, F, T), (U, F, U), (U, T, U), (U, U, T) \} \leq \mathbb{3}^{3}$. Then $(T, T, T) \notin \mathscr{A}(M)$ since $\mathscr{A}(M) = \{ (U, F, F), (F, F, T), (T, T, F), (F, F, U)\}$. However $M$ is not atomic since $(U, T, F)$ can only be written as join of atoms $(U, F, F)$ and $(T, T, F)$ but the $\oplus$ of these atoms is not defined.
\end{remark}

%%%%%%%%%%%%%%%%%%%%%%%%%% Finite atomic $C$-algebras %%%%%%%%%%%%%%%%%%%%%%%%%%

\subsection{Finite atomic $C$-algebras} \label{SectionFinAtomCAlg}

 We establish a characterisation of all finite atomic $C$-algebras. First we establish some results on the existence of $\oplus$ in $M$ where $M$ is an arbitrary $C$-algebra with $T, F, U$.

\begin{proposition} \label{PropExisVeeBarM}
 Consider $M \leq \mathbb{3}^{X}$ for some set $X$ and let $\alpha_{i} \in M$ for $1 \leq i \leq N$ be represented by the pairs of sets $(A_{i}, B_{i})$ respectively. Then $\displaystyle \bigoplus_{i = 1}^{N} \alpha_{i}$ exists if and only if $A_{i} \cap (A_{j} \cup B_{j})^{c} = \emptyset$ for all $i, j \in I$.
\end{proposition}

\begin{proof}
 If $A_{i} \cap (A_{j} \cup B_{j})^{c} = \emptyset$ for all $i, j \leq N$ then
  \begin{align*}
  \alpha_{1}(x) \vee \alpha_{2}(x) \vee \cdots \vee \alpha_{N}(x) & = \begin{cases}
                                                                T, & \text{ if } x \in A_{1}; \\
                                                                U, & \text{ if } x \in (A_{1} \cup B_{1})^{c}; \\
                                                                \alpha_{2}(x) \vee \alpha_{3}(x) \vee \cdots \vee \alpha_{N}(x), & \text{otherwise.} \\
                                                                \end{cases} \\
                                                            & = \begin{cases}
                                                            T, & \text{ if } x \in A_{1}; \\
                                                            U, & \text{ if } x \in (A_{1} \cup B_{1})^{c}; \\
                                                            T, & \text{ if } x \in A_{2}; \\
                                                            U, & \text{ if } x \in (A_{2} \cup B_{2})^{c}; \\
                                                            \alpha_{3}(x) \vee \cdots \vee \alpha_{N}(x), & \text{otherwise.} \\
                                                            \end{cases}
 \end{align*}
 Note that the well-definedness of this expression follows from the fact that $A_{i} \cap (A_{j} \cup B_{j})^{c} = \emptyset$ so that we do not have $x \in A_{1} \cap (A_{2} \cup B_{2})^{c}$ or $x \in A_{2} \cap (A_{1} \cup B_{1})^{c}$. This process yields the following:
  \begin{equation*}
  \alpha_{1}(x) \vee \alpha_{2}(x) \vee \cdots \vee \alpha_{N}(x) = \begin{cases}
                                                                T, & \text{ if } x \in \bigcup A_{i}; \\
                                                                U, & \text{ if } x \in \bigcup (A_{i} \cup B_{i})^{c}; \\
                                                                F, & \text{otherwise}
                                                                \end{cases}
 \end{equation*}
 which is well-defined and establishes that the join is independent of the order of the elements. Consequently $\bigoplus_{1 \leq i \leq N} \alpha_{i}$ exists and can be expressed as follows:
 \begin{equation*}
  \bigoplus_{1 \leq i \leq N} \alpha_{i}(x) = \begin{cases}
                                        T, & \text{ if } x \in A_{i} \text{ for some } 1 \leq i \leq N; \\
                                        U, & \text{ if } x \in (A_{i} \cup B_{i})^{c} \text{ for some } 1 \leq i \leq N; \\
                                        F, & \text{otherwise.}
                                        \end{cases}
 \end{equation*}

 Conversely, suppose if possible that $x \in A_{i} \cap (A_{j} \cup B_{j})^{c}$ for some $x \in X$ and some $i, j \leq N$ where $i \neq j$. Then $(\alpha_{i} \vee \alpha_{j})(x) = T \vee U = T$ while $(\alpha_{j} \vee \alpha_{i})(x) = U \vee T = U$, a contradiction to the fact that $\bigoplus_{1 \leq i \leq N} \alpha_{i}$ is defined. The result follows.
\end{proof}

\begin{proposition} \label{PropIJExisVeeBar}
 Let $\alpha_{i} \in M$ for $i \in I$ where $(\emptyset \neq ) \text{ } I$ is finite, such that $\displaystyle \bigoplus_{i \in I} \alpha_{i}$ exists. Consider $\emptyset \neq J \subseteq I$. Then $\displaystyle \bigoplus_{j \in J} \alpha_{j}$ exists.
\end{proposition}

\begin{proof}
 Consider $M \leq \mathbb{3}^{X}$ for some set $X$. Let $\alpha_{i}$ be identified with the pair of sets $(A_{i}, B_{i})$ for each $i \in I$. Since $\bigoplus_{i \in I} \alpha_{i}$ exists, using \pref{PropExisVeeBarM} we have $A_{i_{1}} \cap (A_{i_{2}} \cup B_{i_{2}})^{c} = \emptyset$ for all $i_{1}, i_{2} \in I$. Thus $A_{j_{1}} \cap (A_{j_{2}} \cup B_{j_{2}})^{c} = \emptyset$ for all $j_{1}, j_{2} \in J$ so that $\bigoplus_{j \in J} \alpha_{j}$ exists.
\end{proof}

We now arrive at the main result in this section.

\begin{theorem} \label{ThmAtomicAda}
 Let $M$ be a finite $C$-algebra with $T, F, U$. $M$ is atomic if and only if $M$ is an ada.
\end{theorem}

\begin{proof}
 ($\Leftarrow$) In view of \rref{RemFinAda3X} we have $M$ is isomorphic to $\mathbb{3}^{X}$ for some finite set $X$. Using \tref{ThmFin3XAtomic} we establish that $M$ is atomic.

 ($\Rightarrow$) If possible let $M$ be atomic and $M$ not be an ada. Then $M \lneq \hat{M}$ where $\hat{M}$ is the enveloping ada of $M$. Consider $\hat{M} \leq \mathbb{3}^{X}$ as adas for some finite set $X$. Thus $M \leq \hat{M} \leq \mathbb{3}^{X}$ as $C$-algebras.

 Since $M \lneq \hat{M}$ there exists $\gamma \in M$ such that $\gamma^{\downarrow} \notin M$. Therefore there exists $x_{1} \in X$ such that $\gamma(x_{1}) = T$ since otherwise $\gamma^{\downarrow} = {\bf F} \in M$. Further, there exists $x_{2} \in X$ such that $\gamma(x_{2}) = U$ since otherwise $\gamma^{\downarrow} = \gamma \in M$, a contradiction. Hence $\gamma$ can be identified with the pair of sets $(A, B)$ where $A \neq \emptyset \neq (A \cup B)^{c}$.

 Since $M$ is atomic there exist $\alpha_{i}$ where $i \in I$ ($I$: finite) and $\alpha_{i} \in \mathscr{A}(M) \cap M_{\#}$ and $\beta_{j}$ where $j \in J$ ($J$: finite) and $\beta_{j} \in \mathscr{A}(M) \cap M_{\#}^{c}$ such that $$\gamma = (\bigoplus \alpha_{i}) \oplus (\bigoplus \beta_{j}).$$
 It is clear that each $\alpha_{i}$ can be identified with the pair of sets $(A_{i}, A_{i}^{c})$ and that each $\beta_{j}$ can be identified with the pair of sets $(\emptyset, B_{j}^{c})$ where $A_{i}, B_{j} \subseteq X$. In other words $\alpha_{i} = \alpha_{_{T, A_{i}}}$ and $\beta_{j} = \beta_{_{U, B_{j}}}$.

 Since we have ascertained that $A \neq \emptyset \neq (A \cup B)^{c}$ we have $I \neq \emptyset \neq J$. Since $\oplus$ is defined, using \pref{PropExisVeeBarM} we have $A_{i} \cap (\emptyset \cup (B_{j})^{c})^{c} = A_{i} \cap B_{j} = \emptyset$ for all $i \in I$ and $j \in J$. Further, $\bigcup A_{i} = A$.

 Since $I$ is finite we have $\bigoplus \alpha_{i} \in M_{\#} \subseteq M$. Also $\bigoplus \alpha_{i} = \gamma^{\downarrow}$ since $\gamma^{\downarrow}$ is represented by the pair of sets $(A, A^{c})$ and $\bigoplus \alpha_{i}$ is represented by the pair of sets $(\bigcup A_{i}, (\bigcup A_{i})^{c})$. Thus $\gamma^{\downarrow} \in M$ which is a contradiction. The result follows.
\end{proof}

%%%%%%%%%%%%%%%%%%%%%%%%%%%%%%%%%%%%%%%%% C-sets %%%%%%%%%%%%%%%%%%%%%%%%%%%%%%%%%%%%%%%%%

\section{$C$-sets and closure operators} \label{SecCsetsandclosureop}

In \cite{panicker16}, the authors introduced the notion of a $C$-set to study an axiomatization of {\tt if-then-else} that included models of possibly non-halting programs and tests, where the tests were drawn from a $C$-algebra. Given a $C$-algebra, there is an inherent {\tt if-then-else} operation on it, which aids us in studying structural properties of $C$-algebras.

\begin{definition}
 Let $S_{\bot}$ be a pointed set with base point $\bot$ and $M$ be a $C$-algebra with $T, F, U$. The pair $(S_{\bot}, M)$  equipped with an action \[\_\; [\_\; , \_] : M \times S_{\bot} \times S_{\bot} \rightarrow S_{\bot}\] is called a \emph{$C$-set} if it satisfies the following axioms for all $\alpha, \beta \in M$ and $s, t, u, v \in S_{\bot}$:
 \begin{align}
  U[s, t] & = \bot \label{EC1} & \text{($U$-axiom)} \\
  F[s, t] & = t \label{EC6} & \text{($F$-axiom)} \\
  (\neg \alpha)[s, t] & = \alpha[t, s] \label{EC5} & \text{($\neg$-axiom)} \\
  \alpha[\alpha[s, t], u] & = \alpha[s, u]  \label{EC3} & \text{(positive redundancy)} \\
  \alpha[s, \alpha[t, u]] & = \alpha[s, u] \label{EC4} & \text{(negative redundancy)} \\
  (\alpha \wedge \beta)[s, t] & = \alpha[\beta[s, t], t] \label{EC7} & \text{($\wedge$-axiom)} \\
  \alpha[\beta[s, t], \beta[u, v]] & = \beta[\alpha[s, u], \alpha[t, v]] \label{EC2} & \text{(premise interchange)} \\
  \alpha[s, t] = \alpha[t, t] & \Rightarrow (\alpha \wedge \beta)[s, t] = (\alpha \wedge \beta)[t, t] \label{EC8} & \text{($\wedge$-compatibility)}
 \end{align}
\end{definition}

Motivating examples of $C$-sets include $\big( \mathcal{T}_{o}(X_{\bot}), \mathbb{3}^{X}\big)$ with the action
\begin{equation} \label{FunctionalAction}
 \alpha[f, g](x) =
 \begin{cases}
  f(x), & \text{ if } \alpha(x) = T; \\
  g(x), & \text{ if } \alpha(x) = F; \\
  \bot, & \text{ otherwise }
 \end{cases}
\end{equation}
and $(S_{\bot}, \mathbb{3})$ with the action
 \begin{equation*}
  \alpha[a, b] =
  \begin{cases}
   a, & \text{ if } \alpha = T; \\
   b, & \text{ if } \alpha = F; \\
   \bot, & \text{ if } \alpha = U.
  \end{cases}
 \end{equation*}

\begin{example} \label{e-m-m}
Let $M$ be a $C$-algebra with $T, F, U$. By treating $M$ as a pointed set with base point $U$, the pair $(M, M)$ is a $C$-set under the following action for all $\alpha, \beta, \gamma \in M$:
 $$\alpha[\beta, \gamma] = (\alpha \wedge \beta) \vee (\neg \alpha \wedge \gamma).$$
Hereafter, the action of the $C$-set $(M, M)$ will be denoted by double brackets $\_\; \llbracket \_ \;, \_ \rrbracket$.
\end{example}

%%%%%%%%%%%%%%%%%%%%%%%%%%%%%% Theorems on C-sets %%%%%%%%%%%%%%%%%%%%%%%%%%%%%%

Henceforth, unless explicitly mentioned otherwise, an arbitrary $C$-algebra with $T, F, U$ is always denoted by $M$ and an arbitrary $C$-set by $(S_{\bot}, M)$. In \cite{panicker16}, the authors show the following result.

\begin{proposition} \label{PropCset}
For all $\alpha \in M$ we have $\alpha[\bot, \bot] = \bot$.
\end{proposition}

Given a $C$-algebra $M$ with $T, F, U$, there is always another $C$-algebra ensconced in it.

\begin{proposition}
 The set $M_{\#}^{c}$ is a $C$-algebra under the induced operations of $M$.
\end{proposition}

\begin{proof}
 Let $\alpha \in M_{\#}^{c}$. Then $\neg \alpha \vee \neg \neg \alpha = \neg \alpha \vee \alpha = \alpha \vee \neg \alpha \neq T$ and so $\neg \alpha \in M_{\#}^{c}$. Let $\alpha, \beta \in M_{\#}^{c}$. Then considering $M \leq \mathbb{3}^{X}$ for some set $X$ it follows that there exists $x \in X$ such that $\alpha(x) = U$. Therefore $(\alpha \wedge \beta)(x) = \alpha(x) \wedge \beta(x) = U \wedge \beta(x) = U$ and so $\alpha \wedge \beta \in M_{\#}^{c}$. Thus $M_{\#}^{c}$ is closed under $\neg$ and $\wedge$, and therefore under $\vee$. The result follows.
\end{proof}

\begin{remark}
 Note that although $M_{\#}^{c}$ is a $C$-algebra under the induced operations of $M$, it is not closed under the constants $T$ and $F$, and is therefore not a subalgebra of $M$ (with $T, F, U$). It is therefore natural to consider $\overline{M_{\#}^{c}} = M_{\#}^{c} \cup \{ T, F \}$, which is clearly closed with respect to $T, F, U$.
\end{remark}

We include concepts related to closure operators in the following.

\begin{definition}
Given a set $X$, a function $C : \powerset(X) \rightarrow \powerset(X)$ is termed a \emph{closure operator} on X if for all $A, B \subseteq X$ it satisfies the following:
\begin{align}
 A & \subseteq C(A) \tag{extensive} \\
 C^{2}(A) & = C(A) \tag{idempotent} \\
 A \subseteq B & \Rightarrow C(A) \subseteq C(B) \tag{isotone}
\end{align}
A subset $A \subseteq X$ is called a \emph{closed subset} if $C(A) = A$. The set of all closed sets of $X$ ordered by set inclusion $\subseteq$ is a partially ordered set and is denoted by $L_{C}$.
\end{definition}

\begin{theorem} \label{ThmClosOpCompLatt}
If $C$ is a closure operator on $X$ then $L_{C}$ forms a complete lattice.
\end{theorem}

\begin{definition}
An \emph{algebraic closure operator} on $X$ is a closure operator $C$ such that for every $A \subseteq X$ we have $C(A) = \bigcup \{ C(B) : B \subseteq A \text{ and } B \text{ is finite} \}$.
\end{definition}

\begin{definition}
An element $a$ of lattice $L$ is \emph{compact} if whenever $a \leq \vee A$ for some subset $A$ of $L$ for which $\bigvee A$ exists, then there exists a finite subset $B \subseteq A$ such that $a \leq \bigvee B$. A lattice is \emph{compactly generated} if every element is the $\sup$ of compact elements. An \emph{algebraic lattice} is one that is both complete and compactly generated.
\end{definition}

\begin{theorem}
If $C$ is an algebraic closure operator on $X$ then $L_{C}$ is an algebraic lattice, and the compact elements of $L_{C}$ are precisely the closed sets $C(A)$ where $A$ is a finite subset of $X$.
\end{theorem}

\begin{definition}[\cite{maclane71}]
Let $A$ and $B$ be posets and $F : A \rightarrow B$ and $G : B \rightarrow A$ be two antitone functions. The pair $(F, G)$ is said to be an \emph{antitone Galois connection} if for all $a \in A, b \in B$,
$$b \leq F(a) \Leftrightarrow a \leq G(b).$$
\end{definition}

\begin{theorem} \label{ThmAntiGal}
Given an antitone Galois connection $(F, G)$ of posets $A$ and $B$, the composite functions $FG : B \rightarrow B$ and $GF : A \rightarrow A$ form closure operators and are called the associated closure operators. Further, $FGF = F$ and $GFG = G$.
\end{theorem}

%%%%%%%%%%%%%%%%%%%%%%%%%% Chapter 6 %%%%%%%%%%%%%%%%%%%%%%%%%%

\section{An application of if-then-else}

In \ssref{SectionAnnihilators} we introduce a notion of annihilators in $C$-algebras through the {\tt if-then-else} action. The notion of Galois connection yields a closure operator in terms of annihilator, which in turn, yields closed sets. In \ssref{SectionClosedSet3X} we give a characterisation for the closed sets in the $C$-algebra $\mathbb{3}^{X}$ (cf. \tref{Thm-Closed-3-X}) and show that this collection forms a complete Boolean algebra (cf. \tref{BooleanAlg}). We also obtain a classification of the elements of $\mathbb{3}^{X}$ where the elements of the Boolean algebra $\mathbb{2}^{X}$ form a distinct class (cf. \tref{ThmPartition}).

In this section, unless stated otherwise, $M$ is a $C$-algebra with $T, F, U$.

%%%%%%%%%%%%%%%%%%%%%%%%%%%%%%%%%%%%%%%%% Annihilators %%%%%%%%%%%%%%%%%%%%%%%%%%%%%%%%%%%%%%%%%

\subsection{Annihilators} \label{SectionAnnihilators}

In this section we show that the presence of the {\tt if-then-else} action on the $C$-algebra $M$ delineates a mechanism to define a notion of annihilators akin to the concept of annihilators in modules.

Henceforth we consider the $C$-set $(M, M)$ where $M$ is a $C$-algebra with $T, F, U$ with the action $\alpha \llbracket \beta, \gamma \rrbracket = (\alpha \wedge \beta) \vee (\neg \alpha \wedge \gamma)$. Since $\alpha \llbracket \; \_ \;, \_ \; \rrbracket$ can be treated as a binary operation for each $\alpha \in M$ we define the notion of the annihilator of an element $a \in M$ to be all the binary operations $\alpha \llbracket \; \_ \;, \_ \; \rrbracket$ which map the pair $(a, a)$ to $U$. We state the definition explicitly in \dref{DefAnnElm}. Hereafter we use $\alpha, \beta, \gamma, \delta$ to denote elements of $M$ treated as binary operations while $a, b, c$ is used otherwise. The elements of the $C$-algebra $\mathbb{3}^{X}$ will also be denoted by $\alpha, \beta, \gamma, \delta$. Recall that the constants $T, F, U$ of $M \leq \mathbb{3}^{X}$ are denoted by ${\bf T}, {\bf F}, {\bf U}$ respectively (cf. \nref{Nota3XConstants}).

\begin{definition} \label{DefAnnElm}
 For $a \in M$, $$Ann(a) = \{ \alpha \in M : \alpha \llbracket a, a \rrbracket = U \}.$$
\end{definition}

We overload the notation of $Ann$ in a natural manner to subsets of $M$.

\begin{definition}
 The operator $Ann : \powerset(M) \rightarrow \powerset(M)$ is given by $$Ann(S) = \displaystyle \bigcap_{a \in S} Ann(a).$$ In other words $Ann(S) = \{ \alpha \in M : \text{ for all } a \in S, \alpha \llbracket a, a \rrbracket = U \}$.
\end{definition}

It is desirable that the operator $Ann$ has some fundamental properties. For instance every element $\alpha \in M$ should annihilate $U$. Conversely every element must be annihilated by $U$. In the following we ascertain these and other properties of the operator $Ann$ which may be deemed natural.

\begin{proposition} \label{PropAnnList}
 The following hold in any $C$-algebra with $T, F, U$:
 \begin{enumerate}[\rm(i)]
  \item $Ann(U) = M$.
  \item For any $a \in M$, $U \in Ann(a)$.
  \item For any $a \in M_{\#}$, $Ann(a) = \{ U \}$.
  \item $Ann(M) = \{ U \}$.
  \item $b \in Ann(a) \Leftrightarrow a \in Ann(b)$.
  \item $B \subseteq Ann(A) \Leftrightarrow A \subseteq Ann(B)$.
  \item $A \subseteq B \Rightarrow Ann(B) \subseteq Ann(A)$.
 \end{enumerate}
\end{proposition}

\begin{proof}$\;$
 \begin{enumerate}[(i)]
  \item Using \pref{PropCset} for the $C$-set $(M, M)$ we have $\alpha \llbracket U, U \rrbracket = U$ for each $\alpha \in M$ so that $Ann(U) = M$.
  \item Using \eqref{EC1} on the $C$-set $(M, M)$ we have $U \llbracket a, a \rrbracket = U$ so that $U \in Ann(a)$ for all $a \in M$.
  \item Using \pref{PropAnnList}(ii) it is clear that $\{ U \} \subseteq Ann(a)$. For the reverse inclusion since $M \leq \mathbb{3}^{X}$ for some set $X$, for $a \in M_{\#}$ we have $a(x) \in \{ T, F \}$ for all $x \in X$. Suppose that $\alpha(x_{o}) \in \{ T, F \}$ for some $x_{o} \in X$ so that $(\alpha(x_{o}) \wedge a(x_{o})) \vee (\neg \alpha(x_{o}) \wedge a(x_{o})) \in \{ T, F \}$. However $Ann(a) = \{ \alpha \in M : (\alpha(x) \wedge a(x)) \vee (\neg \alpha(x) \wedge a(x)) = {\bf U} \text{ for all } x \in X \}$, a contradiction. Consequently $\alpha(x) = U$ for all $x \in X$ hence $\alpha = {\bf U}$.
  \item Since $Ann(M) = \displaystyle \bigcap_{a \in M} Ann(a)$, using \pref{PropAnnList}(iii) we have $Ann(M) = \{ U \}$.
  \item Consider $M \leq \mathbb{3}^{X}$ for some set $X$ and $b \in Ann(a)$ so that $(b(x) \wedge a(x)) \vee (\neg b(x) \wedge a(x)) = U$ for all $x \in X$. For $x \in X$ we have the following cases for $(a(x) \wedge b(x)) \vee (\neg a(x) \wedge b(x))$:
      \begin{description}
       \item [$b(x) = T$] Since $(b(x) \wedge a(x)) \vee (\neg b(x) \wedge a(x)) = U$ we have $(T \wedge a(x)) \vee (F \wedge a(x)) = a(x) \vee F = a(x) = U$. Thus $(a(x) \wedge b(x)) \vee (\neg a(x) \wedge b(x)) = (U \wedge b(x)) \vee (U \wedge b(x)) = U$.
       \item [$b(x) = F$] Along similar lines since $(b(x) \wedge a(x)) \vee (\neg b(x) \wedge a(x)) = U$ we have $(F \wedge a(x)) \vee (T \wedge a(x)) = F \vee a(x) = a(x) = U$. Hence $(a(x) \wedge b(x)) \vee (\neg a(x) \wedge b(x)) = (U \wedge b(x)) \vee (U \wedge b(x)) = U$.
       \item [$b(x) = U$] Then $(a(x) \wedge b(x)) \vee (\neg a(x) \wedge b(x)) = (a(x) \wedge U) \vee (\neg a(x) \wedge U) = U$ for $a(x) \in \{ T, F, U \}$.
      \end{description}
     Hence $a \llbracket b, b \rrbracket = U$ so that $a \in Ann(b)$. The converse follows along similar lines.
  \item This follows as a direct consequence of \pref{PropAnnList}(v).
  \item Let $A \subseteq B$, $\beta \in Ann(B)$ and $a \in A$. Since $\beta \in Ann(b)$ for each $b \in B$ and $a \in A \subseteq B$ we have $\beta \in Ann(a)$. Thus $Ann(B) \subseteq Ann(A)$.
 \end{enumerate}
\end{proof}

We have the following result which follows from \pref{PropAnnList}(vi) and \pref{PropAnnList}(vii).

\begin{proposition} \label{PropAntitoneGal}
The pair $(Ann, Ann)$ is an antitone Galois connection.
\end{proposition}

The following are consequences of \tref{ThmAntiGal} and \pref{PropAntitoneGal}.

\begin{corollary} \label{CorClosOp}
The function $Ann^{2} : \powerset(M) \rightarrow \powerset(M)$ is a closure operator.
\end{corollary}

\begin{corollary} \label{PropCubeOp}
 Consider $Ann^{3} : \powerset(M) \rightarrow \powerset(M)$. Then $Ann^{3} = Ann$.
\end{corollary}

Using \tref{ThmClosOpCompLatt} and \cref{CorClosOp} we have the following.

\begin{corollary} \label{CorClosedSets}
The collection of closed sets $\mathfrak{I} = \{ I \subseteq M : Ann^{2} (I) = I \}$ forms a complete lattice.
\end{corollary}

The closed sets in $\mathfrak{I}$ have the following property.

\begin{proposition}
 Let $I \in \mathfrak{I}$ such that $I \neq M$. Then $I \cap M_{\#} = \emptyset$.
\end{proposition}

\begin{proof}
 Suppose that there exists $a \in I \cap M_{\#}$ where $I \in \mathfrak{I}$ such that $I \neq M$. For any $\alpha \in Ann(a)$ we have $\alpha = U$ using \pref{PropAnnList}(iii) since $a \in M_{\#}$. Further, since $Ann(I) = \displaystyle \bigcap_{a \in I} Ann(a)$ we have $Ann(I) = \{ U \}$. Using \pref{PropAnnList}(i) we have $Ann^{2}(I) = Ann(U) = M$. It follows that $I = Ann^{2}(I) = M$ since $I \in \mathfrak{I}$, which is a contradiction. Thus $I \cap M_{\#} = \emptyset$.
\end{proof}

\begin{remark}
 We established that $Ann^{2}$ is a closure operator. In fact, the following shows that, in general, $Ann^{2}$ need not be an algebraic closure operator.

 Let $M = \mathbb{3}^{\mathbb{N}}$ where $\mathbb{N}$ is the set of natural numbers, viz., $\mathbb{N} = \{ 1, 2, 3, \ldots \}$. Consider the subset $A \subseteq M$ given by $$A = \{ (T, U, U, U, \ldots ), (U, T, U, U, \ldots ), (U, U, T, U, \ldots), \ldots \}.$$ Then
 \begin{align*}
  Ann(A) & = \{ (U, x_{2}, x_{3}, x_{4}, \ldots) : x_{i} \in \{ T, F, U \} \text{ for } i \in \mathbb{N} \setminus \{ 1 \} \} \\
         & \cap \{ (y_{1}, U, y_{3}, y_{4}, \ldots) : y_{i} \in \{ T, F, U \} \text{ for } i \in \mathbb{N} \setminus \{ 2 \} \} \\ & \cap \{ (z_{1}, z_{2}, U, z_{4}, \ldots) : z_{i} \in \{ T, F, U \} \text{ for } i \in \mathbb{N} \setminus \{ 3 \} \} \cap \ldots \\
         & = \{ (U, U, U, U, \ldots) \}
 \end{align*}
 Thus $Ann^{2}(A) = M$.

 If $Ann^{2}$ is an algebraic closure operator then for $i \in I$ where $I$ is some index set, consider $B_{i} \subseteq A$ where $B_{i}$ is finite for each $i$ and $\bigcup Ann^{2}(B_{i}) = Ann^{2}(A) = M$. Then for each $i$, $B_{i}$ comprises elements of the form $\beta(j) = \begin{cases}
                                   T, & \text{ for fixed } j = k_{\beta}; \\
                                   U, & \text{ otherwise.}                                                                                                                           \end{cases}$ \\
 Along similar lines as above it follows that $Ann(B_{i})$ will have elements whose coordinates do not take value $U$ at infinitely many places. Thus all the elements in $Ann^{2}(B_{i})$ will have infinitely many coordinates that take value $U$. If $\bigcup Ann^{2}(B_{i}) = M$ then the element ${\bf T} = (T, T, T, T, \ldots)$ must belong in some $Ann^{2}(B_{i})$, a contradiction since ${\bf T}$ does not take value $U$ in infinitely many coordinates. Thus $Ann^{2}$ is not an algebraic closure operator.
\end{remark}

\subsection{Closed sets of $\mathbb{3}^{X}$} \label{SectionClosedSet3X}

We consider the $C$-algebra $\mathbb{3}^{X}$ and give a characterisation of the closed sets in $\mathfrak{I}$ with respect to operator $Ann^{2}$. To that aim in this section we consider the $C$-algebra in question to be precisely $\mathbb{3}^{X}$ for some set $X$.

\begin{theorem} \label{Thm-Closed-3-X}
Let $I \subseteq \mathbb{3}^{X}$. $I \in \mathfrak{I}$ if and only if there exists $Y \subseteq X$ such that
 \begin{enumerate}
  \item[\rm(P1)] for all $\alpha \in I$, for all $y \in Y$, $\alpha(y) = U$
  \item[\rm(P2)] for all $f : Y^{c} \rightarrow \mathbb{3}$ there exists $\alpha \in I$ such that $\alpha|_{Y^{c}} = f$.
 \end{enumerate}
\end{theorem}

\begin{proof}
 ($\Leftarrow$) Let $I \subseteq \mathbb{3}^{X}$ such that there exists $Y \subseteq X$ which satisfies both the given conditions (P1) and (P2). We show that $Ann^{2}(I) = I$. In view of the fact that $Ann^{2}$ is extensive it suffices to show that $Ann^{2}(I) \subseteq I$.

 Let $\beta \in Ann^{2}(I)$. For $\beta |_{Y^{c}} : Y^{c} \rightarrow \mathbb{3}$ there exists $\alpha \in I$ such that $\alpha |_{Y^{c}} = \beta |_{Y^{c}}$. Moreover, $\alpha(y) = U$ for all $y \in Y$. We show that $\beta(y) = U$ for all $y \in Y$ so that $\beta = \alpha$ from which it follows that $\beta \in I$.

 Suppose if possible, that $\beta(y_{o}) \in \{ T, F \}$ for some $y_{o} \in Y$. Since $\beta \in Ann^{2}(I)$ we have $(\beta \llbracket \gamma, \gamma \rrbracket)(y_{o}) = U$ for all $\gamma \in Ann(I)$ and so $\gamma(y_{o}) = U$ for all $\gamma \in Ann(I)$. Consider $\delta \in \mathbb{3}^{X}$ given by
 \begin{equation*}
  \delta(x) = \begin{cases}
              T, & \text{ if } x \in Y, \\
              U, & \text{ otherwise.}                                                                                                                                                                                                                                                                                                                                                                                                                                                                                                                                                                                                                                                                                                                                                                                                                                                                                                                                                                                                   \end{cases}
 \end{equation*}
 Due to the fact that $\alpha(y) = U$ for all $y \in Y$, for all $\alpha \in I$, we infer that $\delta \llbracket \alpha, \alpha \rrbracket = {\bf U}$ so that $\delta \in Ann(I)$. However $\delta(y_{o}) = T \neq U$, a contradiction. Hence $\beta \in I$ and so $I \in \mathfrak{I}$.

 ($\Rightarrow$) Let $I \in \mathfrak{I}$. Consider the following.
 \begin{align*}
  A & = \{ x \in X : \alpha(x) = U \text{ for all } \alpha \in I \} \\
  B & = \{ x \in X : \alpha(x) \in \{ T, F \} \text{ for some } \alpha \in I \} = X \setminus A.
 \end{align*}
 We show that $Y = A$ is the required set. It is clear that $\alpha(y) = U$ for all $\alpha \in I$ and for all $y \in A$. Let $f : B \rightarrow \mathbb{3}$. Consider its extension $\hat{f} : X \rightarrow \mathbb{3}$ given by the following:
 \begin{equation*}
  \hat{f}(x) = \begin{cases}
                f(x), & \text{ if } x \in B; \\
                U, & \text{ if } x \in A.
               \end{cases}
 \end{equation*}
 Thus $\hat{f} |_{B} = \hat{f} |_{A^{c}} = f$. Let $\beta \in \mathbb{3}^{X}$. It is clear that $$\beta \in Ann(I) \Leftrightarrow \beta(z) = U \text{ for all } z \in B.$$ Consider $\beta \in Ann(I)$. It follows that $(\hat{f} \llbracket \beta, \beta \rrbracket) (z) = U$ for all $z \in B$. Also since $\hat{f}(y) = U$ for all $y \in A$ we have $(\hat{f} \llbracket \beta, \beta \rrbracket) (y) = U$ and so $\hat{f} \llbracket \beta, \beta \rrbracket = U$ from which it follows that $\hat{f} \in Ann^{2}(I) = I$. This completes the proof.
\end{proof}

\tref{Thm-Closed-3-X} equips us with a mechanism to identify the collection of closed sets in $\mathfrak{I}$ with respect to $Ann^{2}$.

\begin{definition} \label{Def-Closed-3-X}
 For $A \subseteq X$ define $I_{A} \subseteq \mathbb{3}^{X}$ by
 \begin{equation} \label{Def-Closed-3-X-Eqn}
  I_{A} = \{ \alpha \in \mathbb{3}^{X} : \alpha(y) = U \text{ for all } y \in A \}.
 \end{equation}
\end{definition}

\begin{proposition} \label{Prop-Closed-3-X}
 $\mathfrak{I} = \{ I_{A} : A \subseteq X \}$.
\end{proposition}

\begin{proof}
 For $A \subseteq X$ consider $I_{A}$ as defined by \eqref{Def-Closed-3-X-Eqn}. It follows in a straightforward manner that $I_{A}$ satisfies (P1) and (P2) of \tref{Thm-Closed-3-X} for $Y = A$ so that $I_{A} \in \mathfrak{I}$.

 Conversely for $I \in \mathfrak{I}$ using \tref{Thm-Closed-3-X} we have $Y \subseteq X$ such that (P1) and (P2) are satisfied. We show that $I = I_{Y}$. Clearly $I \subseteq I_{Y}$ due to (P1). Conversely assume that $\alpha \in I_{Y}$ that is $\alpha(y) = U$ for all $y \in A$. Using (P2) of \tref{Thm-Closed-3-X} we have $\beta |_{Y^{c}} = \alpha |_{Y^{c}}$ for some $\beta \in I$. Property (P1) of \tref{Thm-Closed-3-X} ascertains that $\beta(y) = U$ for all $y \in Y$. It follows that $\alpha = \beta$ so that $\alpha \in I$.
\end{proof}

We make use of the following result to prove \tref{BooleanAlg}.

\begin{lemma} \label{LemmaClosedBool}
 For $A \subseteq X$ and $I_{A} \in \mathfrak{I}$ the following hold:
 \begin{enumerate}[\rm(i)]
  \item $Ann(I_{A}) = I_{A^{c}}$.
  \item $I_{A} \cap I_{B} = I_{A \cup B}$.
  \item $Ann( Ann(I_{A}) \cap Ann(I_{B})) = I_{A \cap B}$.
 \end{enumerate}
\end{lemma}

\begin{proof}$\;$
 \begin{enumerate}[\rm(i)]
  \item Let $\alpha \in Ann(I_{A})$. In view of \dref{Def-Closed-3-X} and \pref{Prop-Closed-3-X} it suffices to show that $\alpha(y) = U$ for all $y \in A^{c}$. For each $y \in A^{c}$ consider $\beta_{y} \in \mathbb{3}^{X}$ given by
 \begin{equation*}
  \beta_{y}(x) = \begin{cases}
                  T, & \text{ if } x = y; \\
                  U, & \text{ otherwise.}
                 \end{cases}
 \end{equation*}
 It is straightforward to see that $\beta_{y} \in I_{A}$ for all $y \in A^{c}$. Thus $\alpha \llbracket \beta_{y}, \beta_{y} \rrbracket = U$ for all $y \in A^{c}$ and so $(\alpha \llbracket \beta_{y}, \beta_{y} \rrbracket) (y) = U$ for all $y \in A^{c}$. Since $\beta_{y}(y) = T$ it follows that $\alpha(y) = U$ for all $y \in A^{c}$.

 For the reverse inclusion consider $\alpha \in I_{A^{c}}$ and $\beta \in I_{A}$. Using \dref{Def-Closed-3-X} and \pref{Prop-Closed-3-X} we have $\alpha(y) = U$ for all $y \in A^{c}$, so that $(\alpha \llbracket \beta, \beta \rrbracket) (y) = U$ for all $y \in A^{c}$. Thus $\beta(y) = U$ for all $y \in A$ so that $(\alpha \llbracket \beta, \beta \rrbracket) (y) = U$ for all $y \in A$. Thus $\alpha \in Ann(I_{A})$ and consequently $Ann(I_{A}) = I_{A^{c}}$.

 \item Consider $\alpha \in I_{A} \cap I_{B}$ and $y \in A \cup B$. It suffices to show that $\alpha(y) = U$. If $y \in A$ then $\alpha(y) = U$ since $\alpha \in I_{A}$. Along similar line $\alpha(y) = U$ if $y \in B$ so that $\alpha \in I_{A \cup B}$.

     For the reverse inclusion consider $\alpha \in I_{A \cup B}$. For $y \in A \subseteq A \cup B$ we have $\alpha(y) = U$ so that $\alpha \in I_{A}$. Proceeding along similar lines we can show that $\alpha \in I_{B}$ from which the result follows.

 \item Using \lref{LemmaClosedBool}(i) and \lref{LemmaClosedBool}(ii) we have $Ann( Ann(I_{A}) \cap Ann(I_{B})) = Ann( I_{A^{c}} \cap I_{B^{c}}) = Ann( I_{A^{c} \cup B^{c}}) = I_{(A^{c} \cup B^{c})^{c}} = I_{A \cap B}$.
 \end{enumerate}
\end{proof}

\begin{theorem} \label{BooleanAlg}
The set $\mathfrak{I}$ of closed sets of $\mathbb{3}^{X}$ with respect to $Ann^{2}$ is a Boolean algebra with respect to the operations
\begin{align*}
 \neg I & = Ann(I) \\
 I_{1} \wedge I_{2} & = I_{1} \cap I_{2} \\
 I_{1} \vee I_{2} & = Ann( Ann(I_{1}) \cap Ann(I_{2}))
\end{align*}
and $\{ {\bf U} \}$ and $\mathbb{3}^{X}$ as the constants $0$ and $1$ respectively. Moreover, $\mathfrak{I} \cong \mathbb{2}^{X}$ and is therefore complete.
\end{theorem}

\begin{proof}
 We rely on the representation of $\mathfrak{I}$ as given in \pref{Prop-Closed-3-X}. In view of \lref{LemmaClosedBool} we show that the operations given as follows define a Boolean algebra on $\mathfrak{I}$:
  \begin{align*}
  \neg I_{A} & = I_{A^{c}} \\
 I_{A} \wedge I_{B} & = I_{A \cup B} \\
 I_{A} \vee I_{B} & = I_{A \cap B}
 \end{align*}

 The verification is straightforward and involves set theoretic arguments.

 Let $A, B \subseteq X$. Using \lref{LemmaClosedBool} we have $I_{A} \wedge I_{B} = I_{A \cup B} = I_{B \cup A} = I_{B} \wedge I_{A}$ and similarly $I_{A} \vee I_{B} = I_{A \cap B} = I_{B \cap A} = I_{B} \vee I_{A}$. Along similar lines the axioms of associativity, idempotence, absorption and distributivity can be verified so that $\langle \mathfrak{I}, \vee, \wedge \rangle$ is a distributive lattice.

 Note that $I_{X} = \{ {\bf U} \}$ while $I_{\emptyset} = \mathbb{3}^{X}$. Using \lref{LemmaClosedBool} we have $I_{A} \wedge I_{X} = I_{A \cup X} = I_{X}$ and $I_{A} \vee I_{\emptyset} = I_{A \cap \emptyset} = I_{\emptyset}$ for all $A \subseteq X$. Also $I_{A} \wedge Ann(I_{A}) = I_{A} \wedge I_{A^{c}} = I_{A \cup A^{c}} = I_{X}$ and $I_{A} \vee Ann(I_{A}) = I_{A} \vee I_{A^{c}} = I_{A \cap A^{c}} = I_{\emptyset}$.

 Hence $\langle \mathfrak{I}, \vee, \wedge, \neg, I_{X}, I_{\emptyset} \rangle$ is a Boolean algebra. It is a straightforward verification to ascertain that the assignment given by $I_{A} \mapsto A^{c}$ from $\mathfrak{I}$ to $\mathbb{2}^{X}$ is a Boolean algebra isomorphism.
\end{proof}

We now give a classification of elements of $M = \mathbb{3}^{X}$ which segregates the elements of $\mathbb{2}^{X} (= M_{\#})$ into one class.

\begin{theorem} \label{ThmPartition}
For each $A \subseteq X$ define $S_{A} = \{ \alpha \in \mathbb{3}^{X} : Ann(\alpha) = I_{A} \}$. The collection $\{ S_{A} : A \subseteq X \}$ forms a partition of $\mathbb{3}^{X}$ in which all the elements of $\mathbb{2}^{X}$ form a single equivalence class.
\end{theorem}

\begin{proof}
 We first show that $S_{A} \neq \emptyset$ for any $A \subseteq X$. To that aim consider $\alpha \in \mathbb{3}^{X}$ given by
 \begin{equation*}
  \alpha(x) = \begin{cases}
               U, & \text{ if } x \in A^{c}; \\
               T, & \text{ otherwise.}
               \end{cases}
 \end{equation*}
 Using \dref{Def-Closed-3-X} and \pref{Prop-Closed-3-X} it is straightforward to verify that $Ann(\alpha) = I_{A}$. Consequently $S_{A} \neq \emptyset$.

 It is self-evident that $\alpha \in S_{A} \cap S_{B}$ is a violation of the well-definedness of $Ann(\alpha)$ from which it follows that $S_{A} \cap S_{B} = \emptyset$ for $A, B \subseteq X$ where $A \neq B$.

 Note that for any $\alpha \in \mathbb{3}^{X}$ we have $Ann(\alpha) \in \mathfrak{I}$ that is $Ann(\alpha) = I_{A}$ for some $A \subseteq X$, since $Ann^{2}(Ann(\alpha)) = Ann^{3}(\alpha) = Ann(\alpha)$ using \cref{PropCubeOp}. Thus $Ann(\alpha) = I_{A}$ for some $A \subseteq X$ so that $\alpha \in S_{A}$. Therefore $\displaystyle \bigcup_{A \subseteq X} \{ S_{A} : A \subseteq X \} = \mathbb{3}^{X}$ and hence the collection $\{ S_{A} : A \subseteq X \}$ forms a partition of $\mathbb{3}^{X}$.

 Further, for $\alpha \in \mathbb{2}^{X}$ we have $Ann(\alpha) = \{ {\bf U } \} = I_{X}$ so that $\alpha \in S_{X}$. Conversely any $\alpha \in S_{X}$ would satisfy $Ann(\alpha) = I_{X} = \{ {\bf U } \}$. If $\alpha(x_{o}) = U$ for some $x_{o} \in X$ then it follows that $\beta \in \mathbb{3}^{X}$ given by
 \begin{equation*}
  \beta(x) = \begin{cases}
             T, & \text{ if } x = x_{o}; \\
             U, & \text{ otherwise.}
             \end{cases}
 \end{equation*}
 satisfies $\beta \llbracket \alpha, \alpha \rrbracket = {\bf U}$ and so ${\bf U} \neq \beta \in Ann(\alpha)$ which is a contradiction. Thus $\alpha(x) \in \{ T, F \}$ for all $x \in X$ from which it follows that $\alpha \in \mathbb{2}^{X}$. Hence the equivalence class $S_{X} = \mathbb{2}^{X}$.
\end{proof}

We conclude this section with some remarks on annihilators.

\begin{remark}$\;$
 \begin{enumerate}[(i)]
  \item The statement $Ann(\alpha) = \{ U \} \Leftrightarrow \alpha \in M_{\#}$ holds in $\mathbb{3}^{X}$ but need not be true in general.

  Consider $M = \{ (T, T), (F, F), (U, U), (F, U), (T, U) \} \leq \mathbb{3}^{2}$ and $(T, U) \in M$. Then $Ann(T, U) = \{ \beta \in M : \beta \llbracket (T, U), (T, U) \rrbracket = (U, U) \}$. Hence for $(x, y) \in Ann(T, U)$ we have $((x \wedge T) \vee (\neg x \wedge T), (y \wedge U) \vee (\neg y \wedge U)) = (U, U)$ and so $(x \vee \neg x, U) = (U, U)$. It follows that $x = U$ and so $Ann(T, U) = \{ (U, U) \}$. However $(T, U) \notin M_{\#}$.

  \item The only closed sets of $M$ are $M$ and $\{ (U, U) \}$. Again note that the collection of closed sets is a Boolean algebra.

  \item For $I \subseteq M$ where $M \leq \mathbb{3}^{X}$ we have $Ann_{M}(I) = Ann_{\mathbb{3}^{X}}(I) \cap M$.

  Let $\alpha \in Ann_{M}(I)$. Clearly $\alpha \in M$ and $\alpha \in Ann_{\mathbb{3}^{X}}(I)$. Conversely suppose $\alpha \in Ann_{\mathbb{3}^{X}}(I) \cap M$. Then it is clear that $\alpha \in Ann_{M}(I)$.

  \item Thus on applying $Ann$ to the previous statement and making appropriate substitutions we have $Ann_{M}^{2}(I) = Ann_{\mathbb{3}^{X}}(Ann_{\mathbb{3}^{X}}(I) \cap M) \cap M$.
 \end{enumerate}
\end{remark}

%%%%%%%%%%%%%%%%%%%%%%%%%%%%%%%%%%%%%%%%% Conclusion %%%%%%%%%%%%%%%%%%%%%%%%%%%%%%%%%%%%%%%%%

\section{Future work} \label{SectionChpt6Concl}

A point of interest would be to enquire whether the representation of elements through atoms by $\oplus$ as defined in this work is unique. Further, by the definition of atomic $C$-algebras proposed by us, we note that $\mathbb{3}^{X}$ is not atomic for infinite $X$. It is therefore desirable to obtain a suitable definition for atomicity so that $\mathbb{3}^{X}$ is atomic for any set $X$. It remains to be seen what characterisation may be achieved for the closed sets of an arbitrary $C$-algebra with $T, F, U$, and whether the closed sets in such a $C$-algebra always form a Boolean algebra.


\begin{thebibliography}{10}

\bibitem{belnap70}
N.~D. Belnap, Jr.
\newblock Conditional assertion and restricted quantification.
\newblock {\em No\^us}, 4(1):1--13, 1970.
\newblock Commentators: W. V. Quine and J. Michael Dunn, Symposia to be held at
  the 68th Annual Meeting of the American Philosophical Association, Western
  Division (St. Louis, Mo., 1970).

\bibitem{bergstra95}
J.~A. Bergstra, I.~Bethke, and P.~Rodenburg.
\newblock A propositional logic with {$4$} values: true, false, divergent and
  meaningless.
\newblock {\em J. Appl. Non-Classical Logics}, 5(2):199--217, 1995.

\bibitem{bochvar38}
D.~A. Bochvar.
\newblock Ob odnom tr\'{e}hznacnom is\v{c}islenii i \'{e}go prim\'{e}n\'{e}nii
  k analiza paradoksov klassi\v{c}\'{e}skogo r\v{a}ssir\'{e}nnogo
  funkcional'nogo is\v{c}isl\'{e}ni\'{a} (in {R}ussian). {\em
  {m}at{\'e}mati{\v{c}}eskij sbornik}, 4: 287--308, 1939. {T}ranslated to
  {E}nglish by {M}. {B}ergmann ``{O}n a three-valued logical calculus and its
  application to the analysis of the paradoxes of the classical extended
  functional calculus''.
\newblock {\em History and Philosophy of Logic}, 2:87--112, 1981.

\bibitem{chang58}
C.~C. Chang.
\newblock Algebraic analysis of many valued logics.
\newblock {\em Trans. Amer. Math. Soc.}, 88:467--490, 1958.

\bibitem{halmos09}
S.~Givant and P.~Halmos.
\newblock {\em Introduction to Boolean Algebras}.
\newblock Springer Science+Business Media, New York, NY, 2009.

\bibitem{guzman90}
F.~Guzm\'{a}n and C.~C. Squier.
\newblock The algebra of conditional logic.
\newblock {\em Algebra Universalis}, 27(1):88--110, 1990.

\bibitem{heyting34}
A.~Heyting.
\newblock Die formalen regeln der intuitionistischen logik, sitzungsberichte
  der preuszischen akademie der wissenschaften, physikalischmathematische
  klasse,(1930), 42--56 57--71 158--169 in three parts.
\newblock {\em Sitzungsber. preuss. Akad. Wiss}, 42:158--169, 1934.

\bibitem{jackson15}
M.~Jackson and T.~Stokes.
\newblock Monoids with tests and the algebra of possibly non-halting programs.
\newblock {\em J. Log. Algebr. Methods Program.}, 84(2):259--275, 2015.

\bibitem{kleene38}
S.~C. Kleene.
\newblock On notation for ordinal numbers.
\newblock {\em The Journal of Symbolic Logic}, 3(4):150--155, 1938.

\bibitem{kleene52}
S.~C. Kleene.
\newblock {\em Introduction to metamathematics}.
\newblock D. Van Nostrand Co., Inc., New York, N. Y., 1952.

\bibitem{lukasiewicz20}
J.~Lukasiewicz.
\newblock On three-valued logic. {R}uch {F}ilozoficzny, 5,(1920), {E}nglish
  translation in {B}orkowski, {L}.(ed.) 1970. {J}an {L}ukasiewicz: {S}elected
  {W}orks, 1920.

\bibitem{maclane71}
S.~Mac~Lane.
\newblock {\em Categories for the working mathematician}.
\newblock Springer-Verlag, New York-Berlin, 1971.
\newblock Graduate Texts in Mathematics, Vol. 5.

\bibitem{manes93}
E.~G. Manes.
\newblock Adas and the equational theory of if-then-else.
\newblock {\em Algebra Universalis}, 30(3):373--394, 1993.

\bibitem{mccarthy63}
J.~McCarthy.
\newblock A basis for a mathematical theory of computation.
\newblock In {\em Computer programming and formal systems}, pages 33--70.
  North-Holland, Amsterdam, 1963.

\bibitem{panicker16}
G.~Panicker, K.~V. Krishna, and P.~Bhaduri.
\newblock Axiomatization of {\tt if-then-else} over possibly non-halting
  programs and tests.
\newblock {\em Int. J. Algebra Comput.}, 27(3):273--297, 2017.

\bibitem{panicker17a}
G.~Panicker, K.~V. Krishna, and P.~Bhaduri.
\newblock Monoids of non-halting programs with tests.
\newblock {\em Algebra Universalis}, 2017.
\newblock To appear. DOI:10.1007/s00012-018-0490-3.

\end{thebibliography}
\end{document}